\documentclass{llncs}
\usepackage{graphicx}
\usepackage{caption}
\usepackage{subcaption}
\usepackage{makeidx}
\usepackage{amssymb}
\usepackage{amsmath}
\usepackage{amsfonts}
\usepackage{amssymb}
\usepackage{amsmath}
\usepackage{amsfonts}
\usepackage[english]{babel}
\usepackage[bottom]{footmisc}
\usepackage{wrapfig}
\usepackage{algorithm}
\usepackage[noend]{algpseudocode}
\usepackage[bottom]{footmisc}
\usepackage{algorithm}
\usepackage{algpseudocode}
\usepackage[colorinlistoftodos,prependcaption,textsize=tiny]{todonotes}
\algnewcommand{\IIf}[1]{\State\algorithmicif\ #1\ \algorithmicthen}
\algnewcommand{\EElse}[1]{\State \algorithmicelse\ #1\ }
\algnewcommand{\EndIIf}{\unskip}

\newcommand{\squishlist}{
   \begin{list}{$\bullet$}
    { \setlength{\itemsep}{2pt}    \setlength{\parsep}{0pt}
      \setlength{\topsep}{5pt}     \setlength{\partopsep}{0pt}
      \setlength{\leftmargin}{1.35em} \setlength{\labelwidth}{1em}
      \setlength{\labelsep}{0.5em} } }

\newcommand{\squishend}{
    \end{list}  }

\begin{document}

\newcommand{\T}{\mathtt{T}}
\newcommand{\E}{\mathbb{E}}

\newcommand{\TransKernel}{\mathcal{T}}
\newcommand{\TransKernelDiscrete}{ \mathcal{T}^{\Delta z}}

\newcommand{\cL}{{\cal L}}

\newtheorem{mydef}{Definition}

\newtheorem{myexam}{Example}

\title{Central Limit Model Checking}% \thanks{This work has been partially supported by }}

\author{Luca Bortolussi \inst{1} \and Luca Cardelli\inst{2,3}  \and Marta Kwiatkowska\inst{2} \and Luca Laurenti \inst{2} }%\}

%\institute{Microsoft Research Cambridge, UK
%\and Department of Computer Science, University of Oxford, UK  \and King's College London \and P\'azm\'any P\'eter Catholic University }

\institute{ Department of Mathematics and Geosciences, University of Trieste \and  Department of Computer Science, University of Oxford, UK \and Microsoft Research Cambridge  }

\maketitle

%syntax highlighting
\newcommand{\varr}[1]{{\color{green!60!black}#1}}
\newcommand{\varc}[1]{{\color{blue!60!black}#1}}
\newcommand{\varl}[1]{{\color{red!60!black}#1}}

%pre-post constraint
\mathchardef\mhyphen="2D
\newcommand{\prepost}[1]{\mathsf{pre\mhyphen post}_{#1}}
\newcommand{\inv}[1]{\mathsf{inv}_{#1}}

\vspace{-2em}
\begin{abstract}
We consider probabilistic model checking for continuous-time Markov chains (CTMCs) induced from Stochastic Reaction Networks (SRNs) against a fragment of Continuous Stochastic Logic (CSL) extended with reward operators.
Classical numerical algorithms for CSL model checking based on uniformisation are limited to finite CTMCs and suffer from exponential growth of the state space with respect to the number of species. On the other hand, approximate techniques such as mean-field approximations and simulations combined with statistical inference are more scalable, but can be time consuming and do not support the full expressiveness of CSL.  
In this paper we employ a continuous-space approximation of the CTMC in terms of a Gaussian process based on the Central Limit Approximation (CLA), also known as the Linear Noise Approximation (LNA), whose solution requires solving a number of differential equations that is quadratic in the number of species and independent of the population size.
We then develop efficient and scalable approximate model checking algorithms on the resulting Gaussian process, where we restrict the target regions for probabilistic reachability to convex polytopes. This allows us to derive an abstraction in terms of a time-inhomogeneous discrete-time Markov chain (DTMC), whose dimension is independent of the number of species, on which model checking is performed. Using results from probability theory, we prove the convergence in distribution of our algorithms to the corresponding measures on the original CTMC.
We implement the techniques and, on a set of examples, demonstrate that they allow us to overcome the state space explosion problem, while still correctly characterizing the stochastic behaviour of the system.
Our methods can be used for formal analysis of a wide range of distributed stochastic systems, including biochemical systems, sensor networks and population protocols.
\end{abstract}

\section{Introduction}

Distributed systems with Markovian interactions can be modelled as continuous-time Markov chains \cite{ethier2009markov}. Examples include randomised population protocols \cite{Angluin2008}, genetic regulatory networks \cite{thattai2001} and biochemical systems evolving in a spatially homogeneous environment, at constant volume and temperature  \cite{gillespie1992rigorous,ethier2009markov}. 
For such systems, stochastic modelling is necessary to describe stochastic fluctuations for low/medium population counts that deterministic fluid techniques cannot capture \cite{ethier2009markov}.

A versatile programming language for modelling the behaviour of Markovian distributed systems is that of \emph{Stochastic Reaction Networks (SRNs)}, which induce CTMCs under certain mild restrictions.
Computing the probability distributions of the species of a SRN over time is achieved by solving the Kolmogorov Equation, also known in the biochemical literature as the Chemical Master Equation (CME) \cite{Kampen1992b}.
 Unfortunately, classical numerical solution methods for computing transient probability based on uniformisation \cite{baier2003model} are often infeasible because of the state space explosion problem, that is, the number of states of the resulting Markov chain grows exponentially with respect to the number of species and may be infinite.
 %formal analysis of the resulting Markov chain is generally infeasible because of the state space explosion problem \cite{cardelli2015stochastic}.
 A more scalable transient analysis can be achieved by employing simulations combined with statistical inference \cite{gillespie1977exact}, but to obtain good accuracy large numbers of simulations are needed, which for some systems can be very time consuming. % or even infeasible.
 
A promising approach, which we explore in this paper, is to instead approximate the CTMC induced by a Stochastic Reaction Network as a \emph{continuous-space} stochastic process by means of the \emph{Central Limit Approximation (CLA)} \cite{ethier2009markov}, also known in statistical physics as the \emph{Linear Noise Approximation (LNA)}. That is, a Gaussian process is derived to approximate the original CTMC \cite{Kampen1992b}. As the marginals of a Gaussian process are fully determined by its expectation and covariances, its solution requires solving a number of differential equations that  is quadratic in the number of species and independent of the population size. As a consequence, the CLA is generally much more scalable than a discrete-state stochastic representation and has been successfully used for analysis of large Stochastic Reaction Networks \cite{cardelli2016stochastic,bortolussi2013model,cardelli2016stochasticHybrid,cardelli2017syntax}.
 However, none of these works enables the computation of complex temporal properties such as global \emph{probabilistic reachability} properties, which quantify the probability of reaching a particular region of the state space in a particular time interval.
 This property is fundamental for verification of more complex temporal logic properties, for example \emph{probabilistic until} properties, where the probability of reaching a certain region within a certain time bound while remaining in another region is quantified. 
 %This property is important not only to analyze %stochastic fluctuations of 
 %biochemical systems, for example to quantify the probability that a particular protein or gene is ever expressed in Gene Regulatory Networks, but is also  fundamental for the verification of more complex temporal logic properties, 
 Such properties can be expressed in Continuous Stochastic Logic (CSL) \cite{aziz1996verifying} or Linear Temporal Logic (LTL) \cite{pnueli1977temporal}, whose formulae are verified by reduction to the computation of the reachability properties \cite{baier2008principles}.
 %or  since model checking for CSL   is reduced to the computation of reachability probabilities.

\subsubsection{Contributions.}
We derive  fast and scalable approximate probabilistic model checking algorithms for CTMCs induced by Stochastic Reaction Networks against a time-bounded fragment of CSL extended with reward operators. 
Our model checking algorithms are numerical and explore a continuous-space approximation of the CTMC in terms of a Gaussian process. One of our key results is a novel scalable algorithm for computing probabilistic reachability for Gaussian processes over target regions of the state space that are assumed to be convex polytopes, i.e. intersections of a finite set of linear inequalities. 
%That is, for a target region specified as a finite set of linear inequalities over the species of a CSRN and a time interval, 
More specifically, for a CTMC approximated as a Gaussian process, the resulting algorithm computes the probability that the system falls in the target region within a specified time interval.
%The LNA gives a semantics that is a Gaussian process. 
Given a set of $k$ linear inequalities, and relying on the fact  that a linear combination of the components of a Gaussian distribution is still Gaussian, we discretize time and space for the $k$-dimensional stochastic process defined by the particular linear combinations. This allows us to derive an abstraction in terms of a time-inhomogeneous \emph{discrete-time Markov chain (DTMC)}, whose dimension is independent of the number of species, since a linear combination is a uni-dimensional entity. The method ensures scalability, as in general we are interested in a small number, i.e., one or at most two, of linear inequalities.
This abstraction is then used to perform model checking of time-bounded CSL properties \cite{Kwiatkowska2007,baier2003model}. To compute such an abstraction, the most delicate aspect is to derive equations for the transition kernel of the resulting DTMC. This is formulated as the conditional probability at the next discrete time step given the system in a particular state.
Reachability probabilities are then computed by making the target set absorbing.
We then extend CSL with the reward operators as in \cite{Kwiatkowska2007}. We derive approximate reward measures for such operators using the CLA, and prove the convergence in distribution of our algorithms to the original measures when the size of the system (number of molecules) tends to infinity. 
We show the effectiveness of our approach on a set of case studies taken from the biological literature, also in cases where existing numerical model checking techniques are infeasible.

A preliminary version of this work has appeared in \cite{bortolussi2016approximation}. This paper extends \cite{bortolussi2016approximation} in several aspects. While in \cite{bortolussi2016approximation} we only consider probabilistic reachability, here we generalise our algorithms to the time-bounded fragment of CSL, which we also extend with reward operators. Furthermore, we prove weak convergence of our algorithms and significantly extend the experimental evaluation.

\subsubsection{Related work.} Algorithms for model checking CSL properties for continuous-time Markov chains have been introduced  and then improved with techniques based on uniformization \cite{baier2000model} (essentially a discretisation of the original CTMC), and reward computation \cite{Kwiatkowska2007}.
%\todo{this introduces uniformisation: Baier C., Haverkort B., Hermanns H., Katoen JP. (2000) Model Checking Continuous-Time Markov Chains by Transient Analysis. In: Emerson E.A., Sistla A.P. (eds) Computer Aided Verification. CAV 2000. Lecture Notes in Computer Science, vol 1855. Springer, Berlin, Heidelberg}
The analysis 
typically involves computing the transient probability of the system residing in a state
at a given time, or, for a model annotated with rewards, the expected reward that can
be obtained.  Despite improvements such as symmetry reduction \cite{heath2008probabilistic}, sliding window \cite{wolf2010solving} and fast adaptive uniformisation \cite{Dannenberg:2015:CCR:2737798.2688907}, their practical use for Stochastic Reaction Networks is severely  hindered by state space explosion \cite{heath2008probabilistic}, which in a SRN grows exponentially with the number of molecules when finite, and may be infinite, in which case finite projection methods have to be used \cite{munsky2006finite}.
As a consequence, approximate but faster algorithms are appealing. The mainstream solution is to rely on simulations combined with statistical inference to obtain estimates \cite{kwiatkowska2011prism,bortolussi2016smoothed}. These methods, however, are still computationally expensive. A recent trend of works explored as an alternative whether estimates  could be obtained by relying on approximations of the stochastic process based on mean-field \cite{bortolussi2012fluid} or linear noise \cite{bortolussi2014stochastic,bortolussi2013model,cardelli2016stochastic}. 
However, CSL and some classes of reward properties, like those considered here, are very challenging. In fact, most approaches consider either local properties of individual molecules \cite{bortolussi2012fluid}, or properties obtained by observing the behaviour of individual molecules and restricting the target region to an absorbing subspace of the (modified) model \cite{bortolussi2013model}. 
The only approach dealing with more general subsets, \cite{bortolussi2014stochastic}, imposes restrictions on the behaviour of the mean-field approximation, whose trajectory has to enter the reachability region in a finite time. Another interesting approach has been developed in \cite{schnoerr2017efficient,milios2017probabilistic}, where model checking of time-bounded properties for CTMCs is expressed as a Bayesian inference problem, and approximated model checking algorithms are derived. However, no guarantees on the convergence of the resulting algorithms is given.

Our approach differs in that 
it is based on the CLA and considers regions defined by polytopes, which encompasses most properties of practical interest. The simplest idea would be to consider the CLA and compute reachability probabilities for this stochastic process, invoking convergence theorems for the CLA to prove the asymptotic correctness. Unfortunately, there is no straightforward way to do this, since dealing with a continuous space and continuous time diffusion process, e.g., Gaussian, is computationally hard, and computing reachability is challenging (see \cite{abate2010approximate}). As a consequence, discrete abstractions are appealing.

\section{Background}
\textbf{Stochastic Reaction Networks.}
A \emph{Stochastic Reaction Network (SRN)} $C=(\Lambda,R)$ is a pair of finite sets, where $\Lambda$ is a set of \emph{species}, $|\Lambda|$ denotes its size, and $R$ is a set of reactions. Species $\lambda \in \Lambda$ interact according to the reactions in $R$. A \emph{reaction} $\tau \in R$ is a triple $\tau=(r_{\tau},p_{\tau},\alpha_{\tau})$,
where $r_{\tau} \in  \mathbb{N}^{|\Lambda|}$ is the \emph{reactant complex}, 
%\emph{source complex}, 
$p_{\tau} \in  \mathbb{N}^{|\Lambda|}$ is the \emph{product complex} and $\alpha_{\tau}: \mathbb{R}_{\geq 0}^{|\Lambda|} \to \mathbb{R}_{\geq 0} $ is the \emph{reaction rate} associated to $\tau$. $r_{\tau}$ and $p_{\tau}$  represent the stoichiometry of reactants and products.
Given a reaction $\tau_1=(  [1,1,0]^T,[0,0,2]^T,\alpha_1 )$, where $\cdot^T$ is the transpose of a vector, we often refer to it as $\tau_1 : \lambda_1 + \lambda_2 \, \rightarrow^{\alpha_1}  \,    2\lambda_3 $.
The \emph{state change} associated to a reaction $\tau$ is defined by $\upsilon_{\tau}=p_{\tau} - r_{\tau}$.  For example, for $\tau_1$ as above, we have $\upsilon_{\tau_1}=[-1,-1,2]^T$. 
A \emph{configuration} or \emph{state} $x \in \mathbb{N}^{|\Lambda|}$ of the system is given by a vector of the number of molecules of each species. 
Given a configuration $x$ then $x_{\lambda_i}$ represents the number of molecules of $\lambda_i$ in the configuration and $\hat x_{\lambda_i}= \frac{x_{\lambda_i}}{N}$ is the \emph{concentration} or \emph{density} of $\lambda_i$ in the same configuration, where $N$ is the population system size, which for molecular systems may represent the volume of the solution, and otherwise it is typically the total population count.

Stochastic Reaction Networks are a versatile programming language used to model stochastic evolution of populations of indistinguishable agents, where the species represent the states of the agents. They are relevant not only for modelling of biochemical systems, such as genetic regulatory networks, molecular signalling pathways and DNA computing circuits, but also certain classes of stochastic distributed systems due to their equivalence to Petri nets \cite{murata1989petri},
Vector Addition Systems (VAS) \cite{karp1969parallel} and distributed population protocols \cite{Angluin2008}.
%\todo{See CAV paper for citations - important to highlight the general usefulness of the work}

\begin{example}
\label{RunningExample}
As a running example we consider the following simple model of gene expression \cite{shahrezaei2008analytical}, where the mRNA is produced by an always active promoter, and  then catalyzes the production of the protein. We have $\Lambda=\{mRNA,Pro\}$ and the following set of reactions $R$:
$$  \to^{0.5} mRNA;\quad mRNA \to^{0.0058\cdot mRNA} mRNA + Pro$$
$$ mRNA\to^{0.0029 \cdot mRNA} \, ;\quad Pro \to^{0.0001\cdot Pro}$$
\end{example}

\subsection{Stochastic Semantics of Stochastic Reaction Networks}\label{subsec-StochSemSRN}
Under the well-mixed assumption \cite{anderson2015models}, a Stochastic Reaction Network $C=(\Lambda,R)$ induces a \emph{discrete-state} Markov process. 
For a reaction $\tau$, $\alpha_{\tau}$ is also called the \emph{propensity rate} of reaction $\tau$ and is a function of the current configuration $x$ of the system, such that $\alpha_{\tau}(x)dt$ is the probability that a reaction event occurs in the next time interval $dt$.
For instance, in case of mass action kinetics, $\alpha_{\tau}(x)=k_{\tau}  \frac{\prod_{i=1}^{|\Lambda|} r_{i,\tau} !   }{N^{|r_{\tau}|-1}}\prod_{i=1}^{|\Lambda|} \binom{x_{\lambda_i}}{r_{i,\tau}}$, where $r_{i,\tau} !$ is the factorial of $r_{i,\tau}$, $|r_{\tau}|=\sum_{i=1}^{|\Lambda|}r_{i,\tau},$ and $x_{\lambda_i}$ is the component of vector $x$ relative to species $\lambda_i$  \cite{Anderson2011}. In this paper we assume $\alpha_{\tau}:\mathbb{R}^{|\Lambda|}_{\geq 0}\to \mathbb{R}_{\geq 0}$ is %a Lipschitz-continuous function. %$\alpha_{\tau}:\mathbb{R}^{|\Lambda|}\to \mathbb{R}_{\geq 0}$ is%%
a real analytic function \cite{bortolussi2012fluid}, that is, a function that locally coincides with its Taylor expansion. This is not restrictive, as it includes all the more commonly used kinetics such as mass action or Hill.  We also require that the SRN satisfies the \emph{density dependent rate} condition\footnote{Note that this condition is not strictly necessary for our results, but guarantees a simpler form for equations \cite{ethier2009markov}.}, that is, for any $\alpha_{\tau}$, there exists a function $\beta_{\tau}:\mathbb{R}^{|\Lambda|}_{\geq 0}\to \mathbb{R}_{\geq 0}$ such that for $x \in \mathbb{R}^{|\Lambda|}_{\geq 0}$ it holds that $\alpha_{\tau}(x)=N \beta_{\tau} (\hat x),$ where $\hat x=\frac{x}{N}$ represents the concentration of the species in $\Lambda$ in configuration $x$.
%Assuming constant volume, the volumetric factor, $N$, is a constant factor, and therefore, to simplify the notation, $N$ is considered embedded inside the coefficient $k$ for any reaction.
 Consequently, a SRN $C=(\Lambda,R)$ is modelled in terms of a \emph{time-homogeneous continuous-time Markov chain} (CTMC) \cite{ethier2009markov} $(X^N(t),t \in \mathbb{R}_{\geq 0})$ with state space $S$ given by the set of possible configurations of the system, where in $X^N$ we made explicit the dependence on the system size $N$.
%\todo{Briefly define reachable states?}
Thus, $X^N(t)$ is a random vector describing the  population count of each species at time $t$. 
Given $X^N$, we denote by $\hat X^N=\frac{X^N}{N}$ the CTMC describing the evolution of the species in $\Lambda$ in terms of concentrations.
The transient evolution of $X^N$, and consequently also of the concentrations $\hat X^N$, is described by the Kolmogorov equations, also called the Chemical Master Equation (CME), namely, a set of differential equations describing the transient evolution of the reachable states $x$.
\begin{definition}{(Kolmogorov Equations)}
\label{Def:CME}
Let $x_0 \in \mathbb{N}^{|\Lambda|}$ be the initial configuration of $X^N$. For $x \in S$, we define $P(x,t|x_0)=Probability(X^N(t)=x \,|\,X^N(0)=x_0)$. $P(x,t|x_0)$ describes the transient evolution of $X^N$, and is the solution of the following system of ordinary differential equations (ODEs):
\begin{equation}\label{CME}
\frac{\mathrm d}{\mathrm d t} \left( P(x,t|x_0)\ \right) = 
	\sum_{\tau \in R} \{ \alpha_{\tau}(x-\upsilon_{\tau})P(x-\upsilon_{\tau},t|x_0)-\alpha_{\tau}(x)P(x,t|x_0)\}. 
\end{equation}
\end{definition}
Solving Eqn \eqref{CME} requires computing the solution of a differential equation for each reachable state. The size of the reachable state space is exponential in the number of the species, and may be infinite. As a consequence, solving the CME is generally feasible only for SRNs with very few species and small molecular populations. This is the so-called state space explosion problem, which strongly limits the applicability of the CME in practice. Finite projection methods have been developed to numerically solve Eqn \eqref{CME} when the state space is not finite \cite{munsky2006finite}. However, they still suffer from the state space explosion problem and are limited to SRNs with few species and moderate population counts.

Often, Eqn \eqref{CME} is a approximated with a deterministic model using fluid techniques \cite{bortolussi2012fluid}, where the concentrations of the species are approximated over time as the solution $\Phi(t)$ of the following set of ODEs, the so-called \emph{rate equations}:
\begin{equation}
% \frac{\mathrm d \Phi(t)}{\mathrm d t} = F(\Phi(t))
\frac{d \Phi(t)}{dt}=F(\Phi(t))=\sum_{\tau \in R}\upsilon_{\tau}\cdot \beta_{\tau}(\Phi(t)),
\label{eq:ODE}
\end{equation}
where in case of mass action kinetics we have $\beta_{\tau}(\Phi(t))=( k_\tau \prod_{i=1}^{|\Lambda|}\Phi_{i}^{r_{i,\tau}}(t))$, for $\Phi_{i}^{r_{i,\tau}}(t)$ the i-th component of vector $\Phi(t)$ raised to the power of $r_{i,\tau}$, i-th component of vector $r_{\tau}$. The initial condition is $\Phi(0)=\frac{x_0}{N}=\hat x_{0}$. Eqn \eqref{eq:ODE} converges to $\hat X^N (t),t \in \mathbb{R}_{\geq 0}$ when $N,$ the system size, tends to infinity \cite{ethier2009markov}. However, Eqn \eqref{eq:ODE} completely neglects the stochastic fluctuations, which may be essential to understand the behaviour of the system being modelled \cite{cardelli2016stochastic}.
\begin{example}
Consider the SRN introduced in Example \ref{RunningExample}. Then, for $t\in \mathbb{R}_{\geq 0}$, we have that $X^N(t)=[X^N_{mRNA}(t),X^N_{Pro}]$ is a random variable describing the number of molecules in the system at time $t$.
Given an initial condition $x_0 \in \mathbb{N}^2_{\geq 0},$ $S$, the state space of $X^N$ is given by the set of states reachable from $x_0$. That is, for any $x \in S$ there is a sequence of reactions $\tau_1,...,\tau_n \in R$ such that $x=x_0+\upsilon_{\tau_1}+...+\upsilon_{\tau_n}.$ Note that the presence of the reaction $\to^{0.5}mRNA $ implies that, in this example, $S$ is not finite. Thus, most of the techniques commonly used for model checking CTMCs would not be directly applicable in this case \cite{Kwiatkowska2007}.
$\hat X^N(t)=[\hat X^N_{mRNA}(t),\hat X^N_{Pro}(t)]=[\frac{ X^N_{mRNA}(t)}{N},\frac{ X^N_{Pro}(t)}{N}]$ describes the evolution of mRNA and Pro in terms of concentrations.

\end{example}

\subsection{Central Limit Approximation}\label{lna-sec}
The \emph{Central Limit Approximation (CLA)}, also called the \emph{Linear Noise Approximation (LNA)}, is a \emph{continuous-space} approximation of the CTMC in terms of  a Gaussian process based on the Central Limit theorem %as a second order approximation of the system size expansion of the CME 
\cite{Kampen1992b,ethier2009markov}.

The CLA at time $t$ approximates the distribution of $X^N(t)$ with the distribution of the random vector $Y^N(t)$ such that:
\begin{equation}
	 X^N(t)\approx Y^N(t) = N\Phi(t) + N^{\frac{1}{2}}G(t)
\label{eq:hypothsis}
\end{equation}
 where $G(t)=(G_1(t),G_2(t),...,G_{|\Lambda|})$ is a random vector, independent of the system size $N$, representing the stochastic fluctuations at time $t$ around $\Phi(t)$, the solution of Eqn \eqref{eq:ODE}.
The probability distribution of $G(t)$ is given by the solution of a linear Fokker-Planck equation \cite{Wallace2012}. As a consequence, for any time instant $t$, $G(t)$ has a multivariate normal distribution whose expected value $\mathbb{E}[G(t)]$ and covariance matrix $cov(G(t))$ are the solution of the following differential equations:
\begin{equation}
		\frac{\mathrm d \mathbb{E}[G(t)]}{\mathrm d t}  = J_F(\Phi(t))\mathbb{E}[G(t)]
\label{LNAEx}
\end{equation}
\begin{equation}
		\frac{\mathrm d cov(G(t)) }{\mathrm d t}  = J_F(\Phi(t))cov(G(t)) + cov(G(t))J^T_F(\Phi(t))+W(\Phi(t))
\label{LNAVar}
\end{equation}
where ${J}_F(\Phi(t))$ is the Jacobian of $F(\Phi(t))$, $J^T_F(\Phi(t))$ its transpose, $ W(\Phi(t))= \sum_{\tau \in R} \upsilon_{\tau} {\upsilon_{\tau}}^T \alpha_{c,\tau}(\Phi(t)) $ and $F_j(\Phi(t))$ the $j$th component of $F(\Phi(t))$. We assume $X^N(0)=x_0$ with probability $1$; as a consequence $\mathbb{E}[G(0)]=0$ and $C[G(0)]=0$, which implies $\mathbb{E}[G(t)]=0$ for every $t$.
The following theorem %, from the work of Ethier and Kurtz \cite{ethier2009markov}, 
illustrates the nature of the approximation using the CLA.
\begin{theorem}[\cite{ethier2009markov}]\label{th:LNA}
Let $C=(\Lambda,R)$ be a SRN, $X^N$ the discrete state space Markov process induced by $C$ and $\hat X^N=\frac{X^N}{N}$. Let $\Phi(t)$ be the solution of Eqn \eqref{eq:ODE}
%$\frac{d \Phi(t)}{dt}=F(\Phi(t))$ 
with initial condition $\Phi(0)=\hat x$ and $G$ be the Gaussian process with expected value and variance given by Eqns \eqref{LNAEx} and \eqref{LNAVar}. Then, for any $t \in  \mathbb{R}_{\geq 0}$ we have:
\begin{equation}
\label{lnconv}
    N^{\frac{1}{2}}\left|\hat X^N(t)-\Phi(t)\right|  \Rightarrow_{N\to \infty} G(t).
\end{equation}
\end{theorem}
In the above, $\Rightarrow_{N\to \infty}$ indicates convergence in distribution as the system size parameter $N$ increases \cite{billingsley2013convergence}.
The CLA is exact in the limit of high populations, but has also been successfully used in many different scenarios showing surprisingly good results \cite{grima2015linear,Wallace2012}.
To compute the CLA it is necessary to solve $O(|\Lambda|^2)$ first order differential equations, and the complexity is independent of the initial number of molecules of each species. %Clearly, it is not necessary to explore the state space to give a stochastic characterization of the network behavior. This is why LNA can be used to handle transient analysis of infinite CTMCs and is much more impervious to state-space explosion than methods based on uniformisation.
Therefore, one can avoid the exploration of the state space that methods based on uniformization rely upon, taking an important step towards scalable stochastic analysis of reaction systems. 

By Eqn \eqref{eq:hypothsis}, we have that the distribution of $Y^N(t)$ is Gaussian with expected value and covariance matrix given by:
\[ \mathbb{E}[Y^N(t)]=N\Phi(t)\]
\[ cov(Y^N(t))=N^{\frac{1}{2}} cov(G(t)) N^{\frac{1}{2}}=N cov(G(t)). \]
Then, the following standard proposition guarantees that a set of linear combinations of the components of $Y^N$ is still Gaussian.
%\todo{add example of linear combination}
\begin{proposition}[\cite{adler2010geometry}]
Let $B \in \mathbb{Z}^{m\times |\Lambda|}$ be a matrix and $Y^N$ a $|\Lambda|-$dimensional Gaussian process. Then, $Z^N=B\cdot Y^N$ is a m-dimensional Gaussian process. For any $t\in \mathbb{R}_{\geq 0}$, we have that $  Z^N(t)$ is characterized by the following mean and covariance:
\begin{equation}
		\mathbb{E}[  Z^N(t)]= B\mathbb{E}[{Y^N}(t)]
		\label{eq:excom}
\end{equation}
\begin{equation}
		 cov(  Z^N(t))= Bcov(Y^N(t))B^T .
		\label{eq:varcom}
\end{equation}
\end{proposition}
\begin{example}
Consider the SRN introduced in example \ref{RunningExample}. According to Theorem \ref{th:LNA} we can associate to $\mathcal{C}$ a Gaussian process $Y^N(t)$ with values in $\mathbb{R}^2.$ Suppose we want to know the distribution of $Z_{mRNA+Pro}^N(t)=Y^N_{mRNA}(t) + Y^N_{Pro}(t)$, where $Y^N_{mRNA}$ and  $Y^N_{Pro}$ are the components of $Y^N$ relative to $mRNA$ and $Pro$. Then, we have that $Z_{mRNA+Pro}^N(t)$ is still Gaussian and with mean and variance given by
$$ E[Z_{mRNA+Pro}^N(t)]=E[Y_{mRNA}^N(t)]+E[Y_{Pro}^N(t)], $$ 
$$  cov(  Z^N_{mRNA + Pro}(t))=[1,1]cov(Y^N(t))[1,1]^T.$$
\end{example}

Thus, $Z^N$ represents the time evolution of $m$ linear combinations of the population counts of the species defined by $B$ over time. Importantly, $Z^N$ is still a Gaussian process, and hence completely characterized by its mean and covariance matrix.
%$Z$ is a Gaussian rocess, as linear combination of the components of a multivariate Gaussian distribution are still Gaussian. 
Note also that the distribution of $\hat Z^N=\frac{Z^N}{N}$ (concentrations) depends on $Y^N$ \emph{only via its mean and covariance}, which are obtained by solving ODEs in Eqns \eqref{LNAEx} and \eqref{LNAVar}.
This is a key feature that we will use to obtain an effective dimensionality reduction in our model checking algorithms.

\section{Continuous Stochastic Logic (CSL)}\label{Section:CSL}
Temporal properties of continuous time Markov chains can be expressed using \emph{Continuous Time Stochastic Logic (CSL)} \cite{aziz2000model}, which can thus be used for the CTMC $X^N$ induced from a SRN $C=(\Lambda,R)$.
We will develop approximate model checking algorithms for CSL based on the CLA.
Since the CLA is correct in the limit of diverging system size $N$,
%In order to compare the behaviour of our approximation algorithms based on the CLA in the limit of diverging $N$ and for this limit to make sense, 
we will define CSL for the \emph{normalized} process $\hat X^N = \frac{X^N}{N},$ as introduced in the previous section. Therefore, we will be working in terms of concentrations instead of population counts. This is not a limitation: if we are interested in a fixed value of $N$, population counts can always be rescaled to population densities, and vice versa, by dividing or multiplying them by ${N}$. In the following, we will thus refer to states and concentrations interchangeably without loss of generality.
%\todo{divide by $N$ or 1/$N$?}
%We now introduce the syntax and semantics of CSL.
%Now, we introduce a probability measure for $\hat X^N$, $Prob$, which will be needed for the definition of CSL. 
%\begin{definition}{(Continuous-time Markov chain)}\label{DefCTMC}
%A (time-homogeneous) continuous-time Markov chain (CTMC) $\hat X^N$ is uniquely defined by a tuple $(S,\mathbf{R},\gamma)$, where $S\subseteq \mathbb{R}^{|\Lambda|}_{\geq 0}$ is a countable set of states, $\gamma:S\to [0.1]$ is the initial distribution,  and $\mathbf{R}:S\times S \to \mathbb{R}_{\geq 0}$ is the \emph{rate transition matrix}. Given $x_1,x_2 \in S,$ the probability that a transition between $x_1$ and $x_2$ occurs in the first $t$ seconds is $1-e^{-t\mathbf{R}(x_2,x_1)}$. Given $x \in S,$ the time spent in $x$ is exponentially distributed with rate $r(x)=\sum_{x_k \neq x}\mathbf{R}(x_k,x),$ and, when a transition occurs, the probability of jumping in state $x'$ is $\frac{\mathbf{R}(x',x)}{r(x)}.$
%\end{definition}

Given a SRN ${C}=(\Lambda,R),$ a \emph{path} of the induced CTMC $\hat X^N$ is defined as $\omega=\hat x_0 t_0 \hat x_1 t_1...$ where $\hat x_k \in \mathbb{R}_{\geq 0}^{|\Lambda|}, t_k \in \mathbb{R}_{\geq 0}$ and for all $k\geq 0$
there exists $\tau \in R$ such that $\beta_{\tau}(\hat x_k)>0$ and $\hat x_k + \frac{\upsilon_{\tau}}{N}=\hat x_{k+1}$, where $\beta_{\tau}$ is the density dependent rate. That is, $\omega$ is an alternating sequence of states (equivalently, concentrations) and residence times in those states.   Let $\Omega$ be the set of all  paths of $\hat X^N$ and $\Omega_{\hat x_0}$ the set of all paths of $\hat X^N$ starting from $\hat x_0$.
Call $\omega(t)$ the state of the path at time $t$, i.e. $\omega(t)=\hat x_n$ where $\sum_{k=0}^{n} t_k \leq t \leq \sum_{k=0}^{n+1} t_k $. %and $\omega @ k=t_k$ is the residence time in the $k-th$ state. 
 Then, a probability measure, Prob, for $\hat X^N$ can be defined using cylinder sets of paths \cite{Kwiatkowska2007}. For further details on the measure-theoretic properties we refer to \cite{baier2003model}.

Since $\hat X^N$ takes values in $\mathbb{R}^{|\Lambda|}_{\geq 0}$, we will work with predicates over concentrations, similarly to how real-time signals are verified in \emph{Signal Temporal Logic (STL)} \cite{maler2004monitoring}, instead of the conventional atomic propositions defined in states of the Markov chain \cite{Kwiatkowska2007}.
\begin{definition}\label{Def:ConvexPredicate}{(Convex Predicate).}
Let $\eta:\mathbb{R}^{|\Lambda|}\to\{\mbox{\emph{true}},\mbox{\emph{false}} \}$ be a predicate. We call $\eta$ a \emph{convex predicate} if there exist $B_1,...,B_m \in \mathbb{Z}^{ |\Lambda|},l_1,...,l_m\in \mathbb{R}, m>0$, such that for $\hat x \in \mathbb{R}^{|\Lambda|}$ it holds that: 
$$\eta(\hat x)=(B_1\cdot  \hat x \leq l_1) \, \wedge . . . \wedge \, (B_m\cdot  \hat x \leq l_m )$$
\end{definition}
Hence, convex predicates are true for concentration $\hat x$ belonging to a, not necessarily bounded, convex polytope. We denote by $\Theta$ the set of all convex predicates with domain in $\mathbb{R}^{|\Lambda|}_{\geq 0}$.

%Given a SRN $C=(\Lambda,R), $ formal analysis of the induced CTMC $\hat X^N$ can be performed against Continuous Time Stochastic logic (CSL) properties \cite{aziz2000model}. %CSL is a widely used formalism for performance evaluation of continuous time stochastic processes \cite{baier2003model}. 
We now define the time-bounded fragment of CSL for SRNs as follows. We do not consider time-unbounded properties because of the nature of the convergence of CLA, which is guaranteed just for finite time.
%\todo{Explain why it has to be bounded}
In Section \ref{SecRew} we extend this fragment with reward operators.
\begin{definition}{(CSL Syntax)}
\label{CSLSyntax}
Given a SRN $C=(\Lambda,R)$, and the induced CTMC $\hat X^N$, we define the syntax of CSL as:
\begin{align*}
    \Psi \, ::=\, \neg \Psi \, | \, \Psi_1 \, \wedge \Psi_2 \, | \,  P_{\sim p}(F^{[t_1,t_2]}\, \eta)   \, | \,  P_{\sim p} (\eta_1 \, U^{[t_1,t_2]}\, \eta_2 )  %| \, Q^{[t_1,t_2]} \eta 
\end{align*}
where $\eta,\eta_1,\eta_2\in \Theta$, $t_1, t_2 \in \mathbb{R}_{\geq 0}$, $ \in [0,1]$  and $\sim\in\{<,>\}$. 
\end{definition}

The above definition slightly differs from the usual definition of CSL in that the reachability ($F^{[t_1,t_2]}$) and until ($U^{[t_1,t_2]}$) operators work directly with predicates over concentrations, rather state labels. Note also that, in Definition \ref{Def:ConvexPredicate}, we do not allow nesting of CSL properties, and we restrict predicates to sets that are convex polytopes. This latter point does not limit the expressivity of the logic. However, it is a fundamental requirement for our model checking algorithms, which allows us to obtain an exponential speed up compared to existing algorithms. 

\begin{example}
Given the SRN $C$ of Example \ref{RunningExample} for $N=100$, the property "is the probability that the concentration of Pro remains below $0.1$  until there is a concentration of mRNA of at least $0.3$ in the first $50$ time units greater than $0.6$?" can be expressed as:
$$P_{>0.6}[(\hat{Pro}<0.1)\, U^{[0,50]}\, (\hat{mRNA}>0.3) ],$$
where with an abuse of notation we call $\hat{Pro}$ and $\hat{mRNA}$ the components of vector $\hat X^N$ relative to species $Pro$ and $mRNA$. Obviously, this property is equivalent to the following one, but expressed on the rescaled process $X^N$:
$$P_{>0.6}[({Pro}<10)\, U^{[0,50]}\, ({mRNA}>30) ].$$
\end{example}

\begin{definition}{(Semantics of CSL)}\label{CSLSemantics}
Let  $\hat X^N$ be the CTMC induced by SRN $C$. Given $\hat x \in \mathbb{R}^{|\Lambda|}_{\geq 0}$, the semantics of CSL is defined as follows: 
\begin{align*}
    \hat X^N,\hat x  \models \neg \Psi \quad   & \leftrightarrow  \quad     \hat X^N,\hat x \not\models \Psi \\
    \hat X^N,\hat x  \models \Psi_1 \wedge \Psi_2  \quad    &\leftrightarrow \quad     \hat X^N  \models \Psi_1 \wedge \hat X^N  \models \Psi_2  \\
    \hat X^N,\hat x \models P_{\sim p}(F^{[t_1,t_2]}\eta)   \quad    &\leftrightarrow \quad     Prob(\exists t \in [t_1,t_2]\, s.t.\, \eta (\omega(t))\, |\, \omega \in \Omega_{\hat x})\sim p \\
    \hat X^N,\hat x \models P_{\sim p} (\eta_1 U^{[t_1,t_2]}\eta_2)  \quad  &
\leftrightarrow  \quad   Prob(\exists t  \in [t_1,t_2]\,  s.t.\, \eta_2 (\omega(t)) \wedge \forall t' \in [0,t)\, \eta_1 (\omega(t'))\,|\,\omega \in \Omega_{\hat x}) )\sim p 
\end{align*}
\end{definition}
Note that the reachability operator can be expressed with the until. For example, $P_{>0.9}[F^{[0,1]}\, mRNA>0]$ is equivalent to $P_{>0.9}[mRNA \geq 0\, U^{[0,1]} \,mRNA>0]$.
Similarly to classical CSL, $\sim$ can be replaced with $=?,$  in the style of quantitative model checking, indicating the probability of satisfaction \cite{hillston2005compositional}. 

Model checking procedures for CTMCs against CSL specifications are well known \cite{Kwiatkowska2007,baier2008principles}.  They reduce to computing the probability of reaching a given set, and hence to solving Eqn \eqref{CME}, albeit resulting in the well known  state space explosion problem.   Here, we explore the usage of the CLA to derive approximate model checking procedures that converge to the original CTMC values but do not suffer from the state space explosion problem, therefore enabling fast stochastic characterization of the system. 

\section{The Reachability Operator}
In this section we define our CLA-based algorithm to verify the probabilistic reachability operator $P_{\sim p}(F^{[t_1,t_2]}\eta),$ which is the key procedure for model checking of more complex CSL properties. As $\eta$ is a convex predicate, in order to check this property, for a convex polytope $A$ defined as $A=\{x \in \mathbb{R}^{|\Lambda|}_{\geq 0}\, s.t.\,\forall i \in \{1,...,m \} (Bx)_i\leq b_i \}$ where $B\in \mathbb{Z}^{m\times |\Lambda|},b \in \mathbb{R}^m$, we need to compute:
$$  P_{reach}^{A}(\hat {x}_0,t_1,t_2)= Prob(\exists t \in [t_1,t_2]\, s.t.\, \omega(t) \in A \, | \, \omega \in \Omega_{\hat {x}_0})  ,$$
where $\Omega_{\hat {x}_0}$ is the set of paths of $\hat X^N$ starting from $\hat x_0$ as defined in Section \ref{Section:CSL}. 
%\todo{Should it be $\hat x_0$? Same for $\Omega$}
We will compute such a probability for $\hat Y^N=\frac{Y^N}{N}$, the CLA of $X^N$ expressed in terms of concentrations, and then show how the computed measure converges to the original process $\hat X^N$, but guaranteeing much greater scalability.
Computing the reachability probability for $\hat Y^N$ is not straightforward, because the system evolves in continuous time and analytic solutions cannot be derived in general. As a consequence, we need to devise numerical algorithms and prove their correctness. Here, we derive a scalable numerical algorithm based on time and space discretization of linear projections of $\hat Y^N$, and, using properties of Gaussian processes, we then prove the convergence of the algorithm to the original measure.

%The (time-bounded) reachability problem for a CRN $C = (\Lambda,R)$, given a region $A$ of the state space, 
%requires the computation of the probability that the CTMC $X(t)$ induced by $C$ enters the region $A$ at some time instant between $t_1$ and $t_2$.  
%\paragraph{Our approach.} 
In order to exploit the CLA, we  first discretize time for the Gaussian process given by the CLA, with a fixed (or adaptive) step size $h$,  which we can do effectively owing to the Markov property and the knowledge of its mean and covariance. As a result, we obtain a \emph{discrete-time, continuous-space}, Markov process with a Gaussian transition kernel. Then, by resorting to state space discretization with parameter $\Delta z>0$, we compute the reachability probability on this new process, obtaining an approximation in terms of time-inhomogeneous discrete-time Markov chain (DTMC)  converging to the CLA approximation  {uniformly, when $h$ and $\Delta z$ go to $0$}.
At first sight, there seems to be little gain, as we now have to deal with a $|\Lambda|$-dimensional continuous state space.
Indeed, for general regions this can be the case. However, if we restrict to regions defined by intersections of linear inequalities (i.e. polytopes), we can exploit  properties of Gaussian distributions (i.e. their closure with respect to linear combinations), reducing the dimension of the continuous space to the number of different linear combinations used in the definition of the linear inequalities (in fact, the same hyperplane can be used to fix both an upper and a lower bound).  As we are generally interested only in one or few projections, the complexity will then be dramatically reduced.

%\subsection{Reachability Problem: Formal Definition}
%Recall that, given a RN $C=(\Lambda,R)$ with initial configuration $x_0$, its stochastic behaviour is described by the CTMC $X^N$.
%A path of $X^N$ is a sequence $\omega=x_0 t_1 x_1 t_1 x_2...$ where $x_i\in \mathbb{N}^{|\Lambda|}$ is a state and $t_i \in \mathbb{R}_{>0} $ is the time spent in the state $x_i$. A path is finite if there is a state $x_k$ that is absorbing. $\omega(t)$ is the state of the path at time $t$. $Path(X^N,x_0)$ is the set of all (finite and infinite) paths of the CTMC starting in $x_0$. We work with the standard probability measure $Prob$ over paths $Path(X^N,x_0)$ defined using cylinder sets \cite{Kwiatkowska2007}.

 %For a simpler presentation, we restrict to a single linear inequality over the species. This still covers many practical scenarios, in particular in systems biology. Next, we show how to generalise the method to regions specified by the intersection of more than one hyperplane, though the complexity of our method will grow exponentially with the number of different hyperplanes.

\subsection{Time Discretization Scheme}
Given $\hat Y^N$, the CLA of $\hat X^N$ expressed in terms of concentrations, and matrix $B\in \mathbb{Z}^{m \times |\Lambda|},$
we introduce an exact time discretization scheme for $\hat Z^N=B \hat Y^N$. For simplicity we assume $m=1,$ but all the results extend to $m>1.$ Fix a small time step $h>0$.  By sampling $\hat Y^N$ at step $h$ and invoking the Markov property,\footnote{The Gaussian process obtained by the Linear Noise Approximation is Markovian, as it is the solution of a linear Fokker-Planck equation (stochastic differential equation) \cite{Kampen1992b}.  } we obtain a \emph{discrete-time Markov process} (DTMP) $\hat{Y}^{h,N}(k) = \hat Y^N(kh)$ on continuous space. Applying the linear projection mapping $\hat Z^N$ to $\hat{Y}^{N}(k)$, and leveraging its Gaussian nature, we obtain a process $\hat {Z}^{h,N}(k) = \hat Z^N(kh)$ which is also a DTMP, though with a kernel depending on time through the mean and variance of $Y^N$. 
%\todo{Should this be $Y^N$ or $\hat Y^N$? see also Prop 2.4.}
\begin{definition}\label{DTMP}
A \emph{(time-inhomogeneous) discrete-time Markov process (DTMP)} $(\hat {Z}^{h,N}(k), k\in [0,I] \subseteq \mathbb{N})$ is uniquely defined by a triple $(S,\mathcal{B}(S),\mathcal{T})$, where $(S,\mathcal{B}(S))$ is a measurable space and $\mathcal{T}:\mathcal{B}(S) \times S \times \mathbb{N}\rightarrow [0,1]$ is a transition kernel such that, for any $z\in S$, $A\in \mathcal{B}(S)$ and $k\in \mathbb{N}$, $\mathcal{T}(A,z,k)$ is the probability that $\hat {Z}^{h,N}(k+1)\in A$ conditioned on $\hat {Z}^{h,N}(k)=z$.
\end{definition}
From Definition \ref{DTMP}, it follows that, for $[0,I]\subseteq \mathbb{N}$, $\hat {Z}^{h,N}$ is a discrete-time stochastic process defined on the sample space given by the product space $\Omega = S^{I+1} $,
endowed  with  the sigma-algebra, $\mathcal{B}(\Omega)$,  generated  by  the
product  topology,   and  with  a  probability  measure $Prob^h$, which is uniquely defined by the transition kernel $\mathcal{T}$ and the initial condition \cite{bertsekas2004stochastic}.

Thus, in order to characterize $\hat {Z}^{h,N}$, we need to compute its transition kernel, $\mathcal{T}$. This is equivalent to computing $f_{\hat Z^N(t+h)|\hat Z^N(t)=\bar{z}}(z)$, i.e. the density function of $\hat Z^N(t+h)$ given the event $\hat Z^N(t)=\bar{z}$.

Consider the joint distribution $(\hat Y^N(t),\hat Y^N(t+h))$, which is Gaussian. Its projected counterpart $(\hat Z^N(t),\hat Z^N(t+h))$ is thus also Gaussian, with covariance function: 
\begin{align*}
cov(\hat Z^N(t),\hat Z^N(t+h))&=  B \, cov(\hat Y^N(t),\hat Y^N(t+h)) \, B^T \\
&=\frac{1}{N} B\,  cov(Y^N(t),Y^N(t+h))\, B^T,
\end{align*}
where $cov(Y^N(t),Y^N(t+h))$ is the covariance function of $Y^N$ at times $t$ and $t+h$. It follows by the closure properties of Gaussian processes that $(\hat Z^N(t+h)|\hat Z^N(t)=\bar{z})$ is Gaussian too, and thus fully characterized by its mean and variance. 
Hence, we need to derive $cov(Y^N(t),Y^N(t+h))$.
From now on, we denote $cov(Y^N(t+h),Y^N(t))=C_{Y^N}(t+h,t)$ and $cov(\hat Z^N(t+h),\hat Z^N(t))=C_{\hat Z^N}(t+h,t)$. 
Following \cite{ethier2009markov}, we introduce the following matrix differential equation:
\begin{equation}\label{AuxEq}
\frac{d \Upsilon(t,s)}{dt}=J_{F}(\Phi(t))\Upsilon(t,s)
\end{equation}
with $t\geq s$ and initial condition $\Upsilon(s,s)=Id$, where $Id$ is the identity matrix of dimension $|\Lambda|$. 
Then, as illustrated in \cite{ethier2009markov}, we have:
\begin{equation}\label{covariance}
C_{Y^N}(t,t+h)=\int_{0}^{t}\Upsilon(t,s)W(\Phi(s))[\Upsilon(t+h,s)]^T ds,
\end{equation}
where $W$ is the matrix introduced in Eqn \eqref{LNAVar}.
This is an integral equation, which has to be computed numerically. To simplify this task, we derive an equivalent representation in terms of differential equations. 
This is given by the following lemma.
\begin{lemma}\label{Cov}
Solution of Eqn \eqref{covariance} is given by the solution of the following differential equations:
\begin{align}\nonumber
    \frac{d C_{Y^N}(t,t+h)}{dt}=&W(\Phi(t))\Psi^T(t+h,t)+J_F(\Phi(t))C_{Y^N}(t,t+h)\\
    &+C_{Y^N}(t,t+h)J^T_F(\Phi(t+h)) \label{CovY}
\end{align}
with initial condition $C_{Y^N}(0,h)$ computed as the solution of:
\begin{equation*}
    \frac{d C_{Y^N}(0,s)}{ds}=C_{Y^N}(0,0+s)J^T_F(\Phi(s)).
\end{equation*}
\end{lemma}
\begin{proof}
Applying the general form of the Fundamental Theorem of Calculus to Eqn \eqref{covariance} with respect to $t$ we get:
\begin{align*}
 \frac{d C_{Y^N}(t,t+h)}{dt}&= \Upsilon(t,t)W(\Phi(t))\Upsilon(t+h,t)^T + \int_{0}^t \frac{d}{dt} (   \Upsilon(s,t)W(\Phi(s))\Upsilon(t+h,s)^T) ds\\
&= Id \cdot W(\Phi(t))\Upsilon(t+h,t)^T +  \int_{0}^t \frac{d \Upsilon(s,t)}{dt}W(\Phi(s))\Upsilon(t+h,s)^T ds\\
&\, \, \, \, \, \, + \int_{0}^t  \Upsilon(s,t)W(\Phi(s))\frac{d \Upsilon(t+h,s)}{dt}^T ds.
\end{align*}
As $\frac{d \Upsilon(t,s)}{dt}=J_F(\Phi(t)) \Upsilon(t,s) $,
we get
\begin{align*}
\frac{d C_{Y^N}(t,t+h)}{dt}=& W(\Phi(t))\Upsilon(t+h,t)^T \\
&+ J_F(\Phi (t)) \int_{0}^t\Upsilon(t,s) W(\Phi(s))\Upsilon(t+h,s)^T ds\\
&+\int_{0}^t  \Upsilon(s,t)W(\Phi(s))\Upsilon(t+h,s)^T  ds J_F(\Phi (t+h))^T.
\end{align*}
By substituting Eqn \eqref{covariance} we have the result.
Similarly, to derive the initial condition $C_{Y^N}(0,h)$ we can apply the Fundamental Theorem of Calculus to Eqn \eqref{covariance}, but with respect to $h$.  
\end{proof}
\noindent
$\Upsilon(t+h,t)$ can be computed by solving Eqn \eqref{AuxEq}. 
Knowledge of $C_{Y^N}(t,t+h)$ allows us to directly compute: $$C_{\hat Z^N}(t,t+h)=\frac{1}{N} B\, C_{Y^N}(t,t+h) \, B^T.$$
Then, by using the law for conditional expectation of a Gaussian distribution, we finally have:

\begin{align*}
\mathbb{E}[\hat Z^N(t+h)|\hat Z^N(t)=\bar{z}]&=\mathbb{E}[\hat Z^N(t+h)]\\
&+C_{\hat Z^N}(\hat Z^N(t+h),Z(t)){C[\hat Z^N(t)]}^{-1}(\bar{z}-\mathbb{E}[\hat Z^N(t)])\\
C[\hat Z^N(t+h)|\hat Z^N(t)=\bar{z}]&=C[\hat Z^N(t+h)]-{C_{\hat Z^N}(t,t+h)}{C_{\hat Z^N}(t,t)}^{-1}{C_{\hat Z^N}(t,t+h)}.
\end{align*}

As the kernel is Gaussian, it is completely determined by its expectation and covariance matrix over time.
Note that the resulting kernel is time-inhomogeneous. The dependence on time is via the mean and covariance of $Y^N$, which are functions of time and define completely the distribution of $Y^N$. The following result, which is a corollary of Theorem 3 in \cite{laurenti2017reachability}, guarantees the correctness of the approximation.
\begin{theorem}
\label{Time Discretization}
Given vector $B\in \mathbb{Z}^{|\Lambda|}$, $b \in \mathbb{R}$, measurable set $A=\{x \in \mathbb{R}^{|\Lambda|}_{\geq 0}\, Bx\leq b \}$ and process $\hat Z^N=B\hat Y^N$ with initial condition $z_0=B \hat x_0 \in \mathbb{R}$, call
$$  P_{reach}^{\hat Y^N,A}(\hat x_0,t_1,t_2)=Prob^{\hat Y^N}(\exists t \in [t_1,t_2 ] \, s.t.\,  \hat Y^N(t) \in A \,|\,\hat Y^N(0)=\hat x_0 ), $$
where $Prob^{\hat Y^N}$ is the Gaussian probabilisty measure of the process $\hat Y^N$.
Further, let $\hat {Z}^{h,N}$ be the DTMP obtained by discretizing $\hat Z^N$ at time step $h>0.$
Then, for $t_1,t_2 \in \mathbb{R}_{\geq 0}$, it holds that  
$$ | P_{reach}^{\hat Y^N,A}(\hat x_0,t_1,t_2)\, - \, Prob^h(\exists k \in [\lfloor \frac{t_1}{h}\rfloor ,\lceil \frac{t_2}{h}\rceil ]\, s.t.\, \hat {Z}^{h,N}(k) \leq b  )|  \to_{h\to 0} 0,$$
uniformly.
\end{theorem}

\subsection{Space Discretization}\label{discretespace}
In order to compute the reachability probability for the DTMP $\hat {Z}^{h,N}$, we discretize its continuous state space into a countable set of non-overlapping cells (regions) of constant size $\Delta z>0$ (except for at most regions of measure $0$, i.e. the boundaries of the cells),  obtaining an abstraction in terms of a discrete-time Markov chain $\hat Z^{\Delta z,h,N}$ with state space $S^{\Delta z}$. Specifically, given $S$, the state space of $\hat {Z}^{h,N}$, $A=\{ x \in \mathbb{R}^{|\Lambda|}\,s.t.\, B x \leq b \}$ the target set for $B\in \mathbb{R}^{|\Lambda|}, b\in \mathbb{R}$, we $S\setminus A$ into a grid of cells of length $2 \Delta z$, where $\Delta z$ defines how fine our space discretization is. For each of the resulting regions we consider a representative point, given by the median of the set. We call the set of representative points $\hat S^{\Delta z} $. Then, we have $S^{\Delta z}=\hat S^{\Delta z} \cup \{z_d^{A} \}, $ where $z_d^{A}$ is the state representing the target set. Theorem \ref{Space Discretization} guarantees that for $\Delta z \rightarrow 0$ the error introduced by the space discretization tends to zero. However, for a fixed $N,$ a possible choice of $\Delta z$ is $\Delta z=\frac{0.5}{N}$, which means that the rescaled process $N \hat {Z}^{\Delta z,h,N}$ takes values in $ \mathbb{Z}$. Nevertheless,   when the population is of the order of hundreds or thousands, it can be beneficial to consider $\Delta z > \frac{0.5}{N}$, since a coarser state space aggregation is reasonable.

Similarly to the previous section (see Definition \ref{DTMP}), as $\hat Z^{\Delta z,h,N}$ is a discrete-time stochastic process, given $[0,I]\subseteq \mathbb{N}$ we can associate to $\hat Z^{\Delta z,h,N}$ a probability space with sample space given by the product space $ ({S}^{\Delta z})^{I+1} $,   and  with  a  probability  measure $Prob^{\Delta z, h}$ uniquely defined by $\mathcal{T}^{\Delta z}$, the transition kernel of $\hat Z^{\Delta z,h,N},$ which is defined as follows. For $z_d',z_d \in \hat S^{\Delta z}$, $\mathcal{T}^{\Delta z}(z_d',z_d,k)$ is defined as: 
\begin{equation}\label{kernel}
   \TransKernelDiscrete(z_d',z_d,k)= \int_{z_d'-\Delta z}^{z_d'+\Delta z} f_{\hat Z^N(hk+h)|\hat Z^N(hk)=z_d}(x) dx,
    \end{equation} 
where $h$ is the discrete time step, assumed to be fixed to simplify the notation. 
For $z_d \in \hat S^{\Delta z}$, we have: 
\begin{equation}\label{kernelUnsafe}
   \TransKernelDiscrete(z_d^{A},z_d,k)= \int_{A} f_{\hat Z^N(hk+h)|\hat Z^N(hk)=z_d}(x) dz,
    \end{equation} 
and for the last case, we have: 
\[\mathcal{T}^{\Delta z}(z_d,z_d^{A},k)=\begin{cases}
    1       & \quad \text{if } \text{$z_d=z_d^{A}$}\\
    0  & \quad \text{otherwise }\\
  \end{cases}.
\]
That is, $z_d^{A}$ is made absorbing. Finally, we define: 
\begin{align*}
P_{reach}^{\Delta z,h,A}(&z_d,t_1,t_2)=\\
&Prob^{\Delta z,h}(\exists k \in [\lfloor \frac{t_1}{h}  \rfloor,\lfloor \frac{t_2}{h}  \rfloor]\, s.t.\, {\hat Z^{\Delta z,h,N}(k)} \in z_d^{A}\, |\, \hat Z^{\Delta z,h,N}(0)=z_d ).
\end{align*}
The following theorem, which is a corollary  of Theorem 2 in \cite{abate2010approximate}, guarantees that the error introduced by the state space approximation tends to zero, decreasing $\Delta z$.
\begin{theorem}
\label{Space Discretization}
Let $\hat {Z}^{h,N}$ be a DTMP, and $\hat Z^{\Delta z,h,N}$ the DTMC obtained by space discretization of $\hat {Z}^{h,N}$ with space discretization step $\Delta z >0.$ Call $z_0$ the initial state of $\hat {Z}^{h,N}$ and $z_{d,0} \in S^{\Delta z}$ the discrete state representing the region containing $z_0.$ 
Then, for $t_1,t_2 \in \mathbb{R}_{\geq 0}$, and measurable set $A \subseteq \mathbb{R},$ 
$$ |Prob^h(\exists k \in [\lfloor \frac{t_1}{h}\rfloor ,\lceil \frac{t_2}{h}\rceil ]\, s.t.\, \hat {Z}^{h,N}(k) \in A | \hat {Z}^{h,N}(0)=z_0 ) - P_{reach}^{\Delta z,h,A}(z_{d,0},t_1,t_2)|  \to_{\Delta z} 0,$$
uniformly.
\end{theorem}

\subsection{Correctness}
To prove the correctness of our numerical algorithm we need to show that, for any measurable set, the reachability measure computed on $\hat X^N$ converges to that computed on $\hat Y^N$. This is guaranteed by the following theorem.
\begin{theorem}
\label{COnvergenceReach}
Let $\mathcal{C}=(\Lambda,R)$ be a SRN with induced CTMC $\hat X^N$ and $\hat Z^{\Delta z,h,N}$ be the DTMC obtained by space and time discretization of $B \hat Y^N$. Assume $\hat X^N(0)=\hat x_0$ and the corresponding initial state for $\hat Z^{\Delta z,h,N}$ is $z_{d,0}$.
Then, for $t_1,t_2 \in \mathbb{R}_{\geq 0}$, $B\in \mathbb{R}^{m\times |\Lambda|}$ and $b \in \mathbb{R}^m$ and $A=\{x \in \mathbb{R}^{|\Lambda|}_{\geq 0}\, s.t.\,\forall i \in \{1,...,m \} (Bx)_i\leq b_i \},$ it holds that:
\begin{align*}
     \lim_{N\to \infty}\lim_{h \to 0}\lim_{\Delta z \to 0}| P_{reach}^A(\hat x_0,t_1,t_2)\, - \, 
     &P_{reach}^{\Delta z,h,A}(z_{d,0},t_1,t_2)|  =0.
\end{align*}
\end{theorem}
%In Theorem \ref{COnvergenceReach},  in order for the scaling to make sense, we need to express the reachability with respect to the normalized process $\hat X = \frac{X^N}{N}.$  As explained in Section \ref{Section:CSL}, this is not restrictive.%, in fact, being the target set a convex polytope, we can easily express the reachability region in terms of the normalized process. For instance, 
%we have that $\frac{B X^N(t)}{N}<b$ if and only if $B X^N(t)< N b.$

The proof of Theorem \ref{COnvergenceReach} is detailed in the Appendix. The main idea is to use Theorems \ref{Space Discretization} and \ref{Time Discretization} to show that the numerical model checking algorithms on the Gaussian process $\hat Y^N$ are sound. Then, we employ Theorem \ref{th:LNA} and the theory of weak convergence to show the convergence in distribution of the reachability measure on $\hat X^N$ to that on $\hat Y^N$. The proof is complicated by the fact that both $\hat Y^N$ and $\hat X^N$ depend on $N$.

%\todo{Algorithm 1: is "Propagate probability" the probability of z? How do you "update" (grow) S? (sometimes you show the update of S, sometimes you don't)} 

%\todo{Algorithm 1: to compute the final sum over I, you must have computed the probability of reaching every element of I. But it is not clear where that happens and where that information is stored (what guarantees that your growing S will include all of I?). You do not seem to consider the states within I that are reached from within I: only the ones that are reached from S/I. Luca L: States within I are absorbing. So, once you need a state within I, you just need to accumulate probability. This is like to say: for every state in I, I always make self transitions with Probability $1$}

\subsection{Computation of Reachability Probabilities}
Our approach for computing reachability probabilities is summarized in Algorithm \ref{General}.
\begin{figure}
	\centering
     \includegraphics[width=1.2\linewidth]{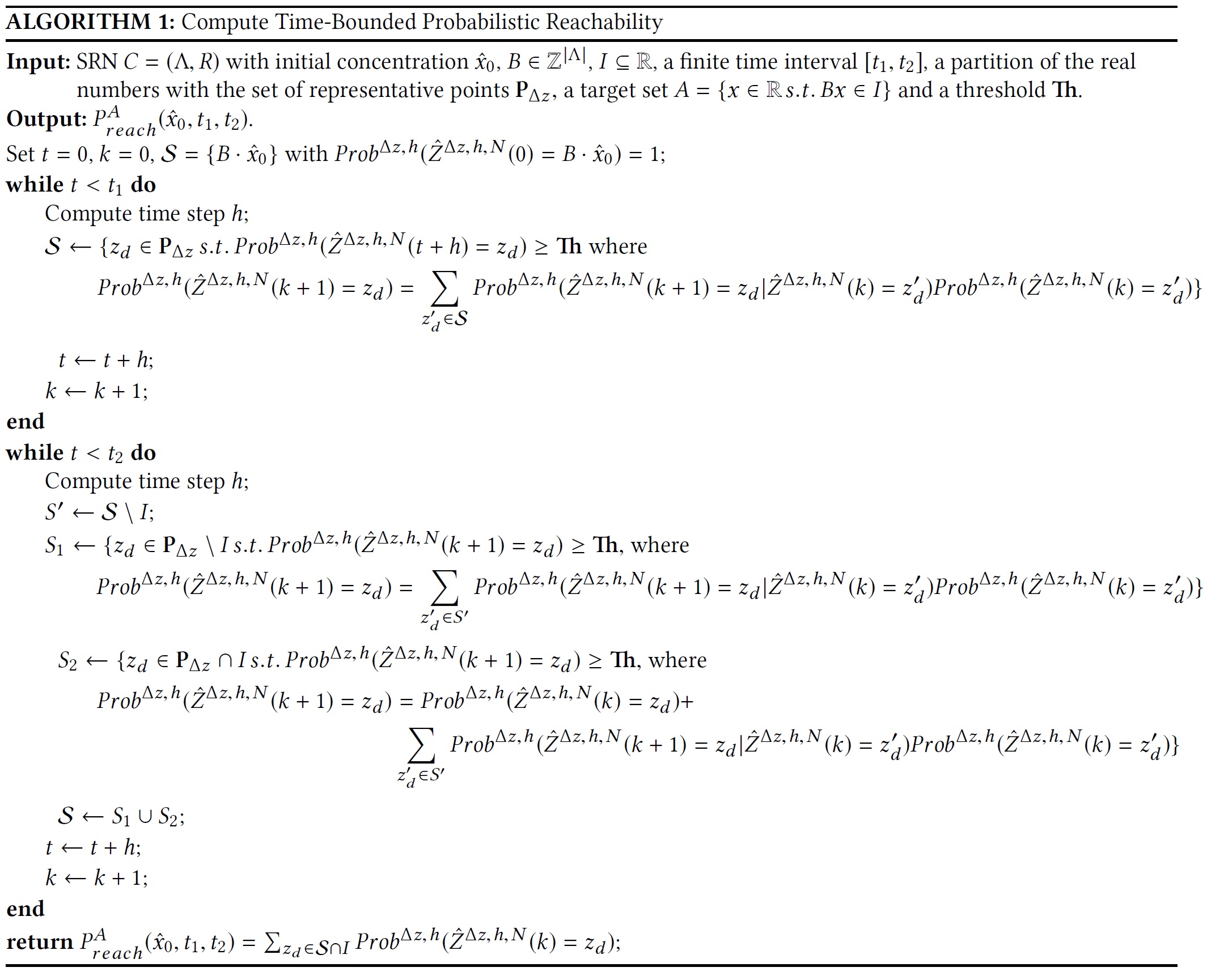}
%\caption{Compute Time-Bounded Probabilistic Reachability }
\label{General}
\end{figure}

In Line $1$, we initialize the system at time $0$. In the context of the algorithm, $\mathcal{S}$ is a set containing the states at a particular time with probability mass greater than the threshold $\mathbf{Th}$.  In our implementation we partition $\mathbb{R}$ with a grid of cells of constant size $\Delta z>0$. Then, for each cell we select a representative point, and denote the set of representative points $\mathbf{P}_{\Delta z}.$ $\mathcal{S}$, for any time $t>0$, will be a subset of this set. $\mathbf{Th}$ equals $10^{-14}$ in all our experiments. The use of a threshold guarantees that the algorithm always terminates in finite time. This introduces a truncation error, which can be easily estimated as shown in \cite{wolf2010solving}. Initially, we have that $\mathcal{S}$ contains only one state $B \cdot \hat x_0$. Then, in Lines $3-7$, we propagate the probability for any discrete step while $t<t_1$, according to classical algorithms for DTMCs \cite{Kwiatkowska2007}. For generality, we assume that the time step $h$ is chosen adaptively, according to the system dynamics.
Propagating probability is possible, as for any $z_d' \in \mathcal{S}$,  $Prob^{\Delta z,h}(\hat Z^{\Delta z,h,N}(k+1)=z_d'|\hat Z^{\Delta z,h,N}(k)=z_d)=\mathcal{T}^{\Delta z}(z_d',z_d,k)$.
From Line $8$ to $15$, we compute the probabilistic reachability, $P_{reach}^A(\hat x_0,t_1,t_2),$ by propagating the probability only for states that are not in A. In fact, states in A are made absorbing.
When we reach $t\geq t_2$, we have that $P_{reach}^A(\hat x_0,t_1,t_2)\approx \sum_{z\in S\cap I}Prob^{\Delta z,h}(\hat Z^{\Delta z,h,N}( \lceil\frac{t_2}{h}\rceil)=z|\hat Z^{\Delta z,h,N}(0)=z_{d,0})$.

%\begin{remark}\label{polytopes}
%Our approach can be easily extended to target regions defined by intersections of finitely many linear inequalities over species. That is, we consider a set of linear predicates $\hat Z^N_j = B_j\cdot X^N(t) \in I_j$, $j=1\ldots,m$ with $m>1$, and ask what is the probability that during a finite time interval we are in a state where each predicate is verified. In order to do that, we can define $B=(B_1,...,B_m)^T\in \mathbb{Z}^{m \times |\Lambda|}$, a matrix where each row is a vector specifying a different linear combination. As a consequence, $\hat Z^N=B\cdot Y^N$ is an $m$ dimensional Gaussian process and all the properties we used for the uni-dimensional case remain valid in this extended scenario. The resulting DTMC $\hat Z^{\Delta z,h,N}$ is $m-$dimensional, but Theorems \ref{Time Discretization} and \ref{Space Discretization} still hold unchanged. 
%\end{remark}

\begin{example}
We return to the SRN introduced in Example \ref{RunningExample}, and, for $N=100,$ we consider the following reachability property:
$$ P_{=?}(F^{[0,T]}\, \hat mRNA>\hat Pro+0.2),\, T\in[0,100]$$
where $=?$, in the style of PRISM \cite{kwiatkowska2011prism} or PEPA \cite{ciocchetta2009bio}, represents the quantitative value of a property. The above formula asks for the probability that, during the first $100$ seconds, the system reaches a state where the mRNA concentration exceeds the protein concentration by more than $0.2$.
In Figure \ref{FigureReachability} we compare the probability value computed by Algorithm \ref{General} with the same property computed on the CTMC $\hat X^N$ using PRISM for different values of $h$. We assume $\Delta _z=\frac{0.5}{N},$ which is justified by the fact that the number of molecules is an integer. 
\begin{figure}
	\centering
     \includegraphics[width=1\linewidth]{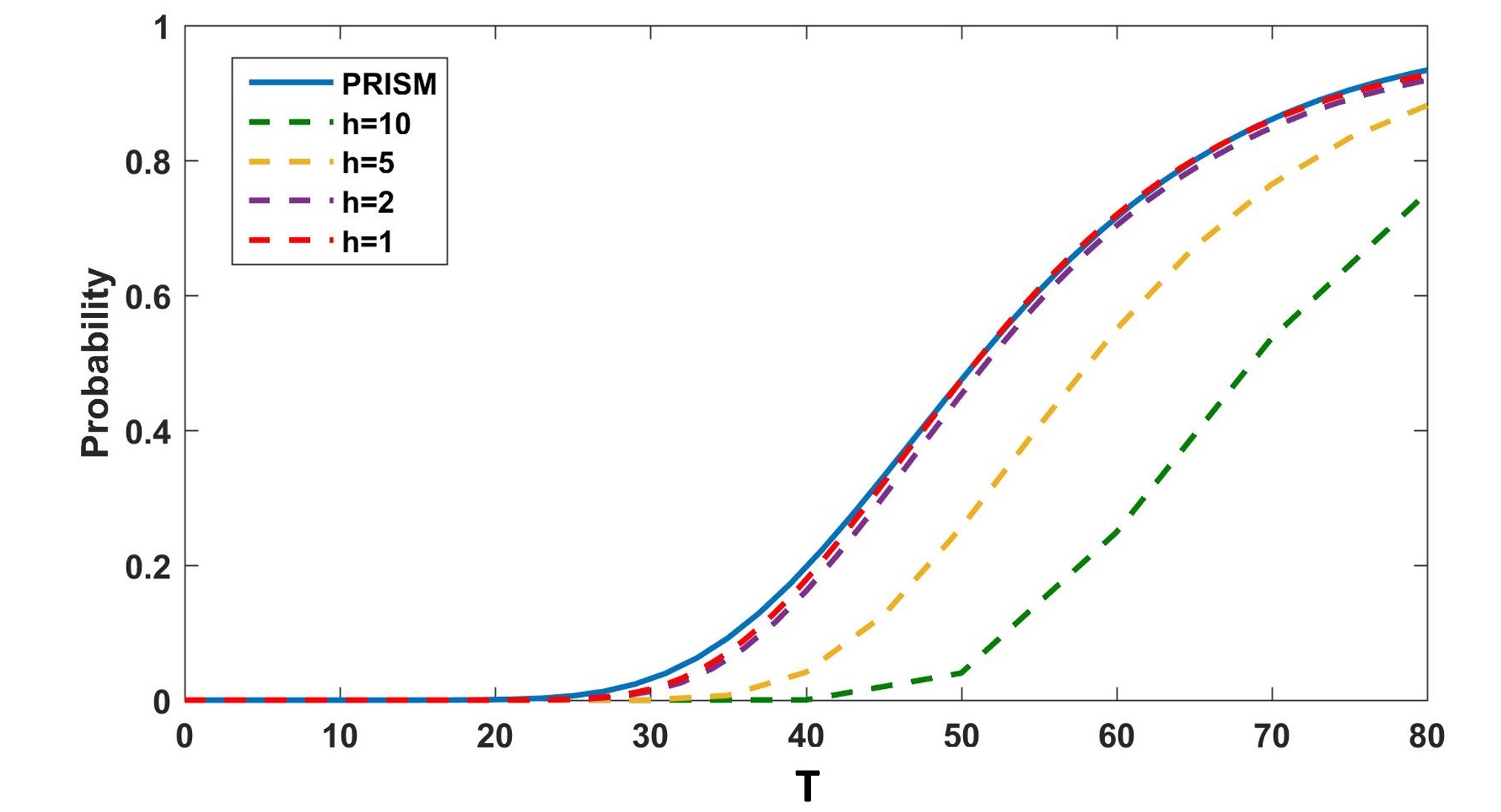}
\caption{Comparison of the evaluation of $P_{=?}(F^{[0,T]} mRNA>Pro+20), T\in[0,100],$ on the CTMC as estimated by PRISM \cite{kwiatkowska2011prism}, and on the CLA approximation for a fixed $\Delta z$ and four different values of $h$. }
\label{FigureReachability}
\end{figure}%
%Note that being the property a linear combination of two jointly normal random variables, the variance of the resulting variable tends to increase linearly with the variance of the two random variables. As a result, a stochastic characterization becomes even more important. 
%Note that, in this example, properties are expressed in terms of densities, while the computation is still performed on the normalized processes. In fact, as explained in Section \ref{Section:CSL}, for a fixed $N$ (here, we considered $N=100$), for the properties of interest in this paper the two representations are equivalent.
\end{example}

{

\section{Until operator}
We show how to generalize the computation of the reachability probabilities of the previous section to  the until operator. For $\hat x \in \mathbb{R}^{|\Lambda|}_{\geq 0}$,
let $\eta_1(\hat x)=B_1 \hat x \leq l_1$ and $\eta_2(\hat x)=B_2 \hat x \leq l_2$, then,  by definition we have: 
\begin{align*}
   \hat X^N, \hat x \models & P_{\sim p}(\eta_1 U^{[t_1,t_2]}\eta_2) \quad \iff \\
 & Prob(\exists t \in [t_1,t_2] \, s.t. \, \eta_2(\omega(t)) \wedge \forall t' \in [0,t), \eta_2(\omega(t))|\omega \in \Omega_{\hat x}),
\end{align*}
where $\Omega_{\hat x}$ is the set of paths of $\hat X^N$ starting in $\hat x$.
To solve this problem we can construct the following stochastic process:
$$\hat Z^N=B \hat Y^N$$
where $B=(B_1,B_2)^T$, and $\hat Y^N$ is the CLA of $\hat X^N$. 
By the properties of multivariate Gaussian distribution, $\hat Z^N$ is still a Gaussian process with mean and covariance matrix given by
$$ \mathbb{E}[\hat Z^N(t)]=B\mathbb{E}[Y^N(t)] \quad C_{\hat Z^N}(t)=\frac{1}{N}B\, C_{Y^N}(t)\, B^T, \, t \in \mathbb{R}_{\geq 0} .$$
Note that $\hat Z^N$ is again a time-inhomogeneous Markov process, as its kernel depends on the statistics of $Y^N$. 
Following the approach of the previous section, we can discretize time and space for $\hat Z^N$, and thus obtain a DTMC $\hat Z^{\Delta z,h,N}$. At this point, the problem reduces to computing the probability for until on the DTMC. Algorithms for computing the resulting measure on a time-inomhogeneous DTMC exist and are well studied \cite{chen2009ltl}. In fact, to compute $P_{\sim p}(\eta_1 U^{[t_1,t_2]}\eta_2),$
we can simply make the regions that do not satisfy $\eta_1$ and those for which $\eta_2$ holds absorbing, and then compute the probability of reaching a region for which $\eta_2$ is satisfied. This can be computed by resorting on Algorithm \ref{General}, as presented in the previous section. Theorem \ref{COnvergenceReach} then guarantees the following proposition.
\begin{proposition}\label{Proposition:Until}
Let $\eta_1(\hat x)=B_1 \hat x \sim l_1$,  $\eta_2(\hat x)=B_2 \hat x \sim l_2$, and $B=\begin{bmatrix}
      B_1          \\ B_2
     \end{bmatrix}$. For $\hat x_0 \in \mathbb{R}^{|\Lambda|}_{\geq 0},$ let $z_{d,0}$ be the state in the state space of $Z^{\Delta z, h, N}$ corresponding to the region containing $B \hat x_0.$ Call
\begin{align*}
 &P_{until}((\hat x_0,\eta_1,\eta_2, \hat X^N,t_1,t_2) )=\\
 &\quad \quad Prob(\exists t \in [t_1,t_2] \, s.t. \, \eta_2(\omega(t)) \wedge \forall t' \in [0,t), \eta_1(\omega(t))\, | \, \omega \in \Omega_{\hat x_0}),\\
&    P_{until}^{\Delta z, h}((z_{d,0},\eta_1,\eta_2, {\hat Z^{\Delta z,h,N}},t_1,t_2) )=\\
    & \quad \quad  Prob^{\Delta z,h}(\exists k \in [\lfloor \frac{t_1}{h}  \rfloor,\lfloor \frac{t_2}{h}  \rfloor] \, s.t. \, \eta_2(Z^{\Delta z, h, N}(k)) \wedge \forall k' \in [0,k-1],\\
    & \quad \quad \quad \quad \quad \quad \quad \quad \quad \quad \quad \quad \quad \quad \quad \quad \eta_1(Z^{\Delta z, h, N}(k'))\,| \,Z^{\Delta z, h, N}(0)=z_{d,0}).
\end{align*} 
 Then, it holds that
\begin{align*}
\label{Until}
 \lim_{N \to \infty}\lim_{h \to 0}\lim_{\Delta z \to 0} |P_{until}&(((\hat x_0,\eta_1,\eta_2, \hat X^N_1,[t_1,t_2]) )-P_{until}^{\Delta z,h}((z_{d,0},\eta_1,\eta_2, {\hat Z^{\Delta z,h,N}},[t_1,t_2]) )|=0.
\end{align*}
%where $\hat Z^{\Delta z,h,N}$ is the DTMC obtained by time and space discretization of $\hat Y^N$ and $,$  for $\eta_1(x)=(B_1 x<l_1)$ and $\eta_2(x)=(B_2 x<l_2),$ $x \in \mathbb{R}.$
\end{proposition}

\begin{example}\label{RunningUntil}
Let us return to the SRN introduced in Example \ref{RunningExample}. We consider the following quantitative property:
$$ P_{=?}[(Pro<10)\, U^{[0,T]}\, ( mRNA > 30)], T \in [0,100], $$
which is satisfied for those paths in which the mRNA population becomes greater than $30$ before the protein population hits $10$ molecules. In Figure \ref{FigureUntil} we evaluate the property for different values of $h$ and fixed $N=100$.
\begin{figure}
	\centering
     \includegraphics[width=1\linewidth]{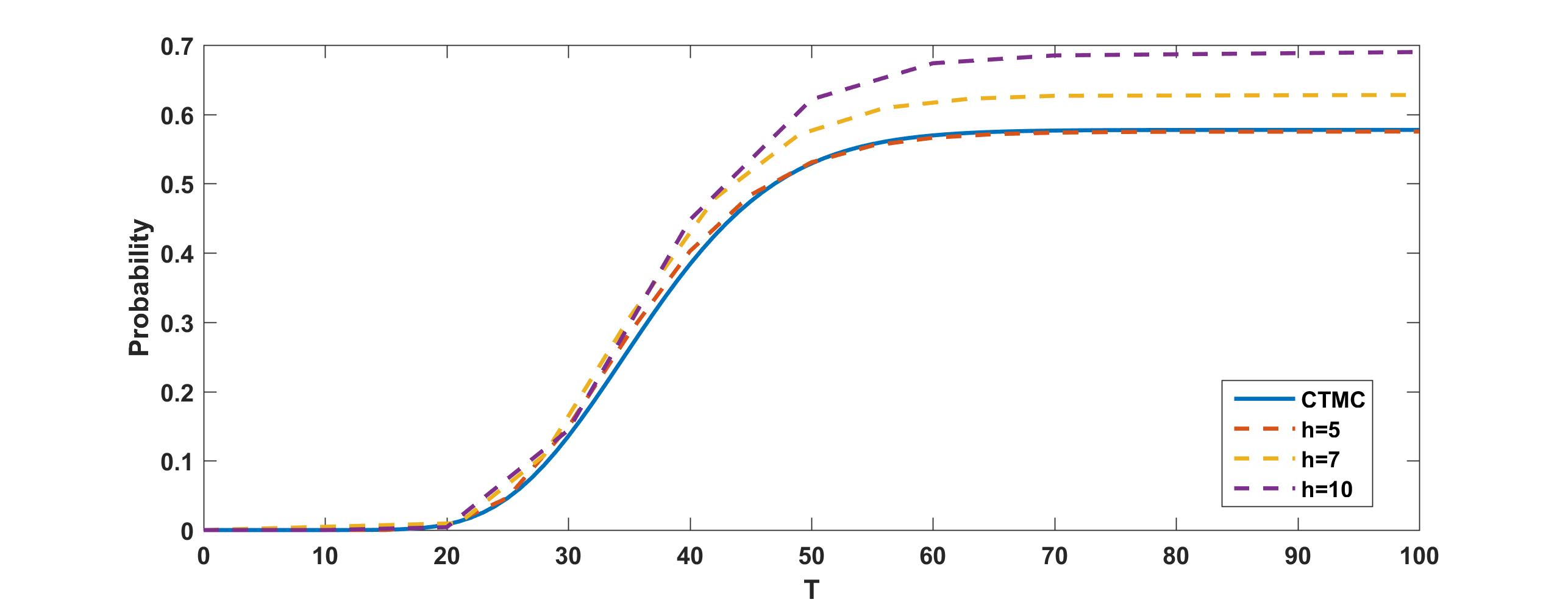}
\caption{Comparison of the evaluation of $P_{[0,T]}[(Pro<10)\,  U^{[0,T]}\,( mRNA > 30)]$ on a CTMC as estimated by PRISM \cite{kwiatkowska2011prism}, and on the CLA approximation for $\Delta z=\frac{0.5}{N}$ and three different values of $h$. }
\label{FigureUntil}
\end{figure}%
Already for $h=5$ the property is surprisingly close to the same measure computed on the original CTMC using uniformization techniques as implemented in PRISM \cite{kwiatkowska2011prism}. Note that here the property is expressed in terms of number of molecules. In fact, as we explained in Section \ref{Section:CSL} for the CSL properties considered in here the two representations are equivalent.
\end{example}

}
\section{Correctness }

The method we present is approximate. In particular, errors are introduced in two ways: by resorting to the CLA and by discretisation of time and space of the CLA.
The quality of these two approximations is controlled by three parameters: (a) $N$, the system size, which influences the accuracy of CLA, (b) $h$, the time step size, and (c) $\Delta z$, the space discretization step, which influences the quality of the approximation of the reachability probability of the CLA.
%Recall that $X^N$ and $\hat Z^{\Delta z,h,N}$ are, respectively, the CTMC induced by a SRN and the DTMC obtained by discretization of the CLA of $X^N$ for a particular $N$. Fix $B \in \mathbb{Z}^{|\Lambda|}$ and A, a set of disjoint closed intervals of reals, and denote by $P_{X^N}(B,t_1,t_2)$ and $P_{\hat Z^{\Delta z,h,N}}(B,t_1,t_2)$, $t_1<t_2$, the reachability probabilities for $\hat Z^{\Delta z,h,N}$ and $X^N$. 
Then, we have the following result.
\begin{theorem}
\label{cor:correctness}
Let $\Psi$ be a CSL formula as defined in Definition \ref{CSLSyntax}, $\hat x \in \mathbb{R}^{|\Lambda|}_{\geq 0}$ and $z_{0,d}$ be the state in $Z^{\Delta z,h,N}$ corresponding to the region containing $\hat x_0$. Then, for $N\to \infty,h\to 0,\Delta z \to 0,$ it holds that:
$$\hat X^N ,\hat x \models \Psi \, \leftrightarrow \, \hat Z^{\Delta z,h,N},z_{d_0} \models \Psi,$$
{almost surely.}
\end{theorem}
\begin{proof}
The proof is by induction on the terms in Definition \ref{CSLSyntax}. The interesting cases are $\Psi=P_{\sim p}(F^{[t_1,t_2]}\eta)$ and $\Psi=P_{\sim p}(\eta_1\,U^{[t_1,t_2]}\, \eta_2)$. Theorem \ref{COnvergenceReach} guarantees that, for  $N\to \infty,h\to 0,\Delta z \to 0,$ the difference between the probability of the above properties computed on $\hat Y^N$, the CLA of $\hat X^N$, and on $\hat X^N$ is equal to $\epsilon\to 0.$
Assume $Prob(\exists t \in [t_1,t_2] \, s.t.\, \eta(\omega(t))|\omega \in \Omega_{\hat x})=q$, and consider the reachability property  $P_{\sim q}(F^{[t_1,t_2]}\eta)$. In this case, no approximation algorithm can guarantee to give the right answer, as the threshold is exactly the value of the reachability property. However, the point $q$ is a set of measure zero with respect to the set of all possible thresholds, which is a subset of the reals. Same reasoning can be applied to the until case.
\end{proof}

\noindent
The convergence stated in Theorem \ref{cor:correctness} means that, since $N$ is fixed for a given SRN, then, even if we have control over $h$ and $\Delta z$, the quality of the approximation depends on how well the CLA approximates the SRN. 
Error bounds would be a viable companion to estimate the committed error, and although these could be extimated for time and space discretization following the approaches in \cite{abate2010approximate,laurenti2017reachability}, we are not aware of any explicit formulation of them for the convergence of the CLA.
 However, experimental results in Section \ref{Exp} show that the error committed is generally limited also for moderately small $N$ and quite large $h$.

\subsection{Complexity}
 Complexity of the method depends on the following: (a) the equations we need to solve, (b) the time step size $h$, and (c) the space discretization step $\Delta z$. 
 Algorithm \ref{General} requires solving  Eqns \eqref{CovY} and \eqref{LNAVar}, that is, a set of differential equations quadratic in the number of species. In fact, solving these equations requires computing $J_F$, Jacobian of $F$. However, the number of equations we need to solve is independent of the number of molecules in the system. This guarantees the scalability of our approach. An important point is that Eqn \eqref{CovY} requires solving Eqn \eqref{covariance} once for each sampling point of the numerical solution of Eqn \eqref{CovY}.
A possible way to avoid this is to consider approximate solutions to Eqn \eqref{covariance}, which are accurate in the limit of $h\rightarrow 0$. However, to keep this  approximation under control, $h$ has to be chosen really small, slowing down the computation. Moreover, for any sample point, Eqn \eqref{covariance} is solved only for a small time interval (between $t$ and $t+h$). As a consequence, in practice, the computational cost introduced in solving Eqn \eqref{covariance} is under control.

A smaller value of $h$ implies that, for a given time interval, we have a greater number of discrete time steps, which can slow down the computation in some cases. The value of $\Delta z$ determines the number of states of the resulting DTMC. However, we stress that we discretize $\hat Z^N(t)$, a uni-dimensional distribution (or $m$-dimensional in the case we have $m>1$ linear inequalities). As a consequence, the number of reachable states with significant probability mass is generally limited and under control. Obviously, if the number of molecules is large and $\Delta z$ extremely small, then this is detrimental to performance.  

\section{Rewards}\label{SecRew}
We extend CSL properties with reward operators as in \cite{Kwiatkowska2007}. As for probabilistic reachability, we will define the reward structure on the normalised process $\hat X^N$. 
Formally, we define the \emph{state reward} function $\rho:\mathbb{R}^{|\Lambda|}\to \mathbb{R}$, which associates a real-valued reward with any point of  the normalised state space of $\hat X^N(t),t\in\mathbb{R}_{\geq 0}$. 
In this paper, we make a few  assumptions about the regularity of $\rho$:
\begin{itemize}
    \item $\rho$ is bounded, i.e. there exists a constant $C>0$ such that $\rho(\hat x)\leq C$ for each $\hat x\in\mathbb{R}^{|\Lambda|}$;
    \item $\rho$ is Lipschitz continuous on $\mathbb{R}^{|\Lambda|}$, i.e. there is a constant $L_\rho$ such that, for each $\hat x,\hat x'\in\mathbb{R}^{|\Lambda|}$, $|\rho(\hat x)-\rho(\hat x')|\leq L_\rho |\hat x- \hat x'|$. 
\end{itemize}
These requirements are important to show the convergence of rewards computed on $\hat X^N$ with the same measure but computed on the normalised CLA $\hat Y^N = \frac{Y^N}{N}$. Moreover, they do not limit the expressiveness of our framework: 
for a fixed $N$, we are interested only in the value of $\rho$ at a finite number of points. Such a function can always be extended to a Lipschitz continuous one. Boundedness is also not problematic, as we can always assume an upper bound on a physically meaningful population size, meaning that we can restrict ourselves to a compact set and define $\rho$ to be constant outside such a set. 

Given a reward structure  $\rho$, we consider the following three kinds of rewards. 
\begin{itemize} 
\item \textbf{Instantaneous Rewards} up to finite time $T$. $\rho_I(\hat x_0,\hat X^N,T)$ is the expectation of $\rho(\hat X^N(T))$, i.e., the expectation of the reward structure at time $T$ over all the trajectories of $\hat X^N$ that start from state $\hat x_0$. More precisely, for $\Omega_{\hat x_0}$, the set of paths of $\hat X^N$ starting from $\hat x_0$:
    \begin{align}
        \rho_I(\hat x_0,\hat X^N,T)=\sum_{x\in \mathbb{R}^{|\Lambda|}}\rho(\hat x)Prob(\omega(T)=\hat x|\omega \in \Omega_{\hat x_0}).  
    \end{align}
    \item \textbf{Cumulative Rewards} up to a finite time $T$. Given $\omega:\mathbb{R}_{\geq 0} \to \mathbb{N}^{|\Lambda|}$, a path of $\hat X^N$, the cumulative reward for $\omega$ is defined as:
    \begin{align*}
        \rho_C(\omega,T)=\int_{0}^T \rho(\omega(t))dt=\sum_{i=1}^{|jumps(\omega)|}&\rho_(\omega(t_{i-1}))(t_i-t_{i-1}) +\\ &\rho_(\omega(T))(T-t_{|jumps(\omega,t)|})
    \end{align*}
    where $jumps(\omega,t)\subset \mathbb{R}_{\geq 0}$ is the set of time instants at which $\omega$ changes state between $0$ and $T$. Then, we define: $$\rho_C(\hat x_0, \hat X^N,T)=\mathbb{E}[\rho_C(\omega,T)\,|\, \omega(0)=\Omega_{\hat x_0}],$$
   where the expectation is intended over the trajectories of $\hat X^N$ starting from state $\hat x_0$
    \item \textbf{Bounded Reachability Rewards.} 
Given a target set $A\in\mathbb{R}^{|\Lambda|}$, for the normalized process $\hat X^N$, define $\rho_{reach}(\hat x_0, X^N,A,T)$, the cumulative reward until we enter the target set $A$ within time $T$. Formally, we can define $\rho_{reach}(\hat x_0, X^N,A,T)$ as the cumulative reward until time  $T$ for the modified process $\bar{X}^N$ where all states in $A$ are made absorbing, and where we consider the modified state rewards:
    $$ \bar{\rho}(\hat x)=\begin{cases}
    0       & \quad \text{if } \hat x \in A\\
    \rho(\hat x)  & \quad \text{otherwise} \\
  \end{cases}  $$ 
\end{itemize}

\begin{remark}
Note that here $\rho$ is a state reward, that is, a function that associates a real value with any given state of the process. An alternative reward structure could be based on the \emph{transition reward} function \cite{bortolussi2015efficient}, which can be used for checking how many times a given reaction fires up to a certain time. However, in the context of SRNs, such a quantity can be easily expressed with an additional species counting how many times a subset of the transitions fire. Then, instantaneous rewards can be used to "read" its value. 
For instance, in Example \ref{RunningExample}, if we want to count the number of times a mRNA molecule is produced, we can consider an additional species $C$ and modify the SRN such that:
$$ \to^{0.5} mRNA + C .$$
Then, for $\hat x \in \mathbb{R}^{|\Lambda|}_{\geq 0}$ we have $\rho(x)=\hat x_C,$ where $\hat x_C$ is the component of vector $\hat x$ relative to species $C$, and $N \rho_I(\hat x_0,\hat X^N,T)$ will give the desired measure at time $T$ for $\hat x_0$, initial concentration of the species.\footnote{The multiplication of $\rho_I$ by $N$ is needed to convert from the normalized process back to the integer population count.}.
\end{remark}

\subsection{Extending CSL with Rewards}
In order to incorporate rewards in our framework, given a SRN $C=(\Lambda,R)$  and the induced CTMC $X^N,$ we extend CSL with the following formulae, whose semantics will depend on the particular reward structure $\rho$:
$$ R_{\sim r}[C_{\rho}^{[\leq T]}] \, | \, R_{\sim r} [I_{\rho}^{=T}]\, | \, R_{\sim r}[F_{\rho}^{\leq T} \eta]$$
where $\eta:\mathbb{R}^{|\Lambda|}\to \{\mbox{\emph{true}}, \mbox{\emph{false}} \}$ is a convex predicate over $\hat X^N$, $T\in \mathbb{R}_{\geq 0},$ $r\in \mathbb{R}_{\geq 0},$ and $\sim \in \{>,<\}$. % \todo{Note: here I need to define the reachable set in terms of the normalised process, but I think we should do this also for reachability, to prove convergence (I think this is used in the proof). }
For $\hat x \in \mathbb{R}^{|\Lambda|}_{\geq 0},$ the semantics of such formulae is as follows:
$$\hat X^N, \hat x \models R_{\sim r}[C_{\rho}^{[\leq T]}] \quad \text{iff} \quad \rho_C(\hat x, \hat X^N,T)\sim r $$
$$\hat X^N, \hat x \models  R_{\sim r} [I_{\rho}^{=T}] \quad \text{iff} \quad \rho_I(\hat x,\hat X^N,T)\sim r $$
$$\hat X^N, \hat x \models R_{\sim r}[F_{\rho}^{\leq T} \eta] \quad \text{iff} \quad \rho_{reach}(\hat x, X^N,A,T)\sim r, $$
where $A=\{\hat x' \in \mathbb{R}^{|\Lambda|}\, s.t.\, \eta(\hat x')\}$. %\todo{nvex se need an example}
%\todo{Do an example - not sure I understand $A$ vs $\eta$}
\subsection{Computing Expectation and Variance Using Reward Operators}
Two of the most common statistics needed when studying stochastic processes are expectation and variance (or covariance) of some set of variables. Suppose we have a CTMC $X^N$ with values in $\mathbb{R}^{|\Lambda|}$, and we want to compute expectation and variance of one of its components $X_i^N$ at time $t$. Then, for $\hat x\in \mathbb{R}^{|\Lambda|}$, we define the following reward structures on the normalised process:
$$ \rho^{size}(\hat x)=
  \begin{cases}
    \hat x_i       & \quad \text{if $\hat x_i< K$} \\
    K  & \quad \text{if } \hat x_i\geq K
  \end{cases}\quad 
  \rho^{size^2}(\hat x)=
  \begin{cases}
\hat    x_i^2      & \quad \text{if } \hat x_i < K\\
    K^2  & \quad \text{if $\hat x_i\geq K$},
  \end{cases} $$
where $K\in \mathbb{R}$ can be any real number, typically an upper bound on the physically admissible population size. For instance, for biochemical processes, we can choose $K$ to be of the order of $10^{80}$, the estimated number of atoms of the universe.  Then, we have 
$$ \mathbb{E}[X_i^N(t)]=N R_{= ?} [I_{\rho_{size}}^{=t}] $$
$$ cov[X_i^N(t)]= N( R_{= ?} [I_{\rho_{size^2}}^{=t}]-(R_{= ?} [I_{\rho_{size}}^{=t}])^2). $$
 $T$ being finite and $K$ any non-negative real,  the above equality holds for any SRN  whose species count remains finite at least for a finite time interval. Note that, as rewards are defined for the normalised process, we need to rescale them back to population counts to compute variance and average of the non-normalised process. 
 %For instance, This is the case in any physical system where $K$ can be taken as greater than the maximum number of molecules in the system.
\subsection{Computing Rewards through Central Limit Approximation}
%\todo{There is a lot of repetition of normalisation, which is covered a the beginning, so this paper needs to be checked if it is necessary to give definitions}
Computing reward properties over $\hat X^N$ is generally not possible because of the state space explosion problem. As a consequence, we compute such properties using $\hat Y^N,$ the CLA of $\hat X^N$. We show that the values computed on $\hat Y^N$ converge (weakly) to those computed on $\hat X^N$. %Since we want to show the convergence for $N\to \infty$, as discussed above we will express rewards in terms of the normalised processes $\hat X^N = \frac{X^N}{N}$, $\hat Y^N = \frac{Y^N}{N}$. This is not limiting. In fact, if we are interested in a fixed value of $N$, we can always either define rewards so that they return the value we are interested in for that given $N$, or  
%rescale back rewards to population counts when feasible. 
We stress again how working in terms of the normalised processes is not a limitation. For instance, consider the reward  for expectation. If we are interested in the expectation of population counts for a fixed $ N$, we can either define the reward for $\hat x$ in the normalised space as $\rho(x) =  N \hat x$, for $N$ fixed, or rescale the computed reward as done in the previous section. 

%. In fact, for most functions, for fixed $N$, we can find a scaling factor to go from a representation to the other.
%For instance, for $\lambda \in \Lambda$ and $\rho(X^N)=(X^N_{\lambda})^2,$ we have that for any $N$,  $\rho(X^N)=N^2 \rho(\frac{X^N}{N}).$ \todo{Need to think a bit more about that. If we have strange non-linear functions, the scaling may be not trivial. In any case, we need rewards to be expressed in terms of concentration for the proofs.}

\subsubsection{Instantaneous Rewards}
Given a reward structure $\rho$, instantaneous rewards can be computed on $\hat Y^N = \frac{Y^N}{N}$ as:
\begin{align*}
\rho^{CLA}_I(\hat x,\hat Y^N,t) \approx \mathbb{E}[\rho(\hat Y^N(t))&] = \int_{K}\rho(x)\mathcal{N}(x|\mathbb{E}[\hat Y^N(t)],cov[\hat Y^N(t)]])dx,
\end{align*}
where $\mathcal{N}(x|\mathbb{E}[\hat Y^N(t)],cov[\hat Y^N(t)]])$ is the normal distribution with mean and covariance matrix respectively, $\mathbb{E}[\hat Y^N(t)],cov[\hat Y^N(t)]]$ for $\hat Y^N(0)=\hat x$. Furthermore,  $K\subseteq \mathbb{R}^{|\Lambda|}$ is a  compact set in which we restrict integration for numerical purposes. The choice of $K$ is such that the error we commit is bounded by a chosen tolerance level.  The following proposition guarantees that $\rho^{CLA}_I(\hat x,\hat Y^N,t)$ converges to $\rho_I(\hat x,\hat X^N, T)$.
\begin{proposition}
\label{ProofInstantaeosReward}
Let $ T \in \mathbb{R}_{\geq 0}$, then it holds that:
$$ \lim_{N\to \infty} \rho_I(\hat x,\hat X^N,T)=\lim_{N\to \infty} \rho^{CLA}_I(\hat x,\hat Y^N,t)$$.
\end{proposition}

\begin{proof}
% Consider the stochastic process $G^N$ defined as

% \[G^N(t) = \sqrt{N}\left(\frac{X^N(t)}{N} - \Phi(t)\right).\]
% Theorem \ref{th:LNA} guarantees that $G^N$ converges (weakly) to the Gaussian Process $G$ on $\mathcal{D}$, the space of $\mathbb{R}^{|\Lambda|}-$valued Cadlag functions (i.e. right continuous functions with left limit \cite{billingsley2013convergence}). 
We want to prove that, for fixed $T>0$, $\mathbb{E}[\rho(\hat X^N)]$ converges to $\mathbb{E}[\rho(\hat Y^N(T))]$ as $N$ tends to infinity.  Using the triangular inequality, it holds that:
\begin{eqnarray*}
|\mathbb{E}[\rho(\hat X^N(T))] &- &\mathbb{E}[\rho(\hat Y^N(T))]| \leq\\  |\mathbb{E}[\rho(\hat X^N))]-\mathbb{E}[\rho(\Phi(T))]| 
&+& | \mathbb{E}[\rho(\Phi(T))]-\mathbb{E}[\rho(\hat Y^N(T))]|,
\end{eqnarray*}
where $\rho(\Phi(T))$ is the reward $\rho$ evaluated on the fluid limit $\Phi(T)$.
Invoking the fluid approximation theorem \cite{bortolussi2013continuous}, it holds that $\hat X^N(T)\Rightarrow_{N\to \infty} \Phi(T)$ (as convergence in probability implies weak convergence). Furthermore,  $\hat Y^N(T)=\frac{G(T)}{\sqrt{N}} + \Phi(T) \Rightarrow_{N \to \infty} \Phi(T)$, as $G$ is independent of $N$ and it has a bounded covariance matrix for each $T$ (which implies convergence in probability).  Therefore, recalling that $\rho$ is bounded and continuous by assumption, both terms on the right hand side of the triangular inequality converge to zero by virtue of the Portmanteau theorem \cite{billingsley2013convergence} stating that, for any continuous and bounded functional $f$ on $\mathcal{D}$, the space of $\mathbb{R}^{|\Lambda|}-$valued Cadlag functions (i.e. right continuous functions with left limit) \cite{billingsley2013convergence},  it holds that $\mathbb{E}[f(X^N)] \rightarrow_{N \to \infty} \mathbb{E}[f(X)]$ whenever $X^N\Rightarrow X$. 
Thus, we can conclude:

$$  \rho_I(\hat x,\hat X^N,T) \to_{N\to \infty} \rho_I^{CLA}(\hat x, \hat Y^N,T).$$
\end{proof}

\begin{example}
\label{Example-InstanRew}
We consider the SRN introduced in Example \ref{RunningExample}. We are interested in knowing the expectation and variance of $mRNA-P$ over time. This can be computed using the following reward structures: $$\rho^{mRNA-P}(\hat x)=\begin{cases}
     {\hat x(mRNA)-\hat x(P)}       & \quad \text{if $\hat x(mRNA)-\hat x(P)< 10^{80}$} \\
    10^{80}  & \quad \text{otherwise} 
  \end{cases}$$ $$\rho^{(mRNA-P)^2}(\hat x)=\begin{cases}
    {(\hat x(mRNA)-\hat x(P))^2}      & \quad \text{if $(\hat x(mRNA)- \hat x(P))^2< 10^{80}$} \\
    10^{80}  & \quad \text{otherwise}
  \end{cases}.$$ Then, we have:  $$\mathbb{E}[X_{mRNA}^N(t)-X_P^N(t)]=N R_{= ?}[I^{=t}_{\rho^{mRNA-P}}], \quad   t \in [0,100],$$
  $$Cov{(X_{mRNA}^N(t)-X_P^N(t))}= N( R_{= ?} [I_{\rho^{(mRNA-P)^2}}^{=t}]-(R_{= ?} [I_{\rho_size}^{=t}])^2), \quad t \in [0,100], $$
  where the rewards are computed on the normalised process $\hat X^N.$
The resulting expectation and variance is plotted in Figure \ref{FigureInstantaneousReward}.
\begin{figure}
	\centering
     \includegraphics[width=1\linewidth]{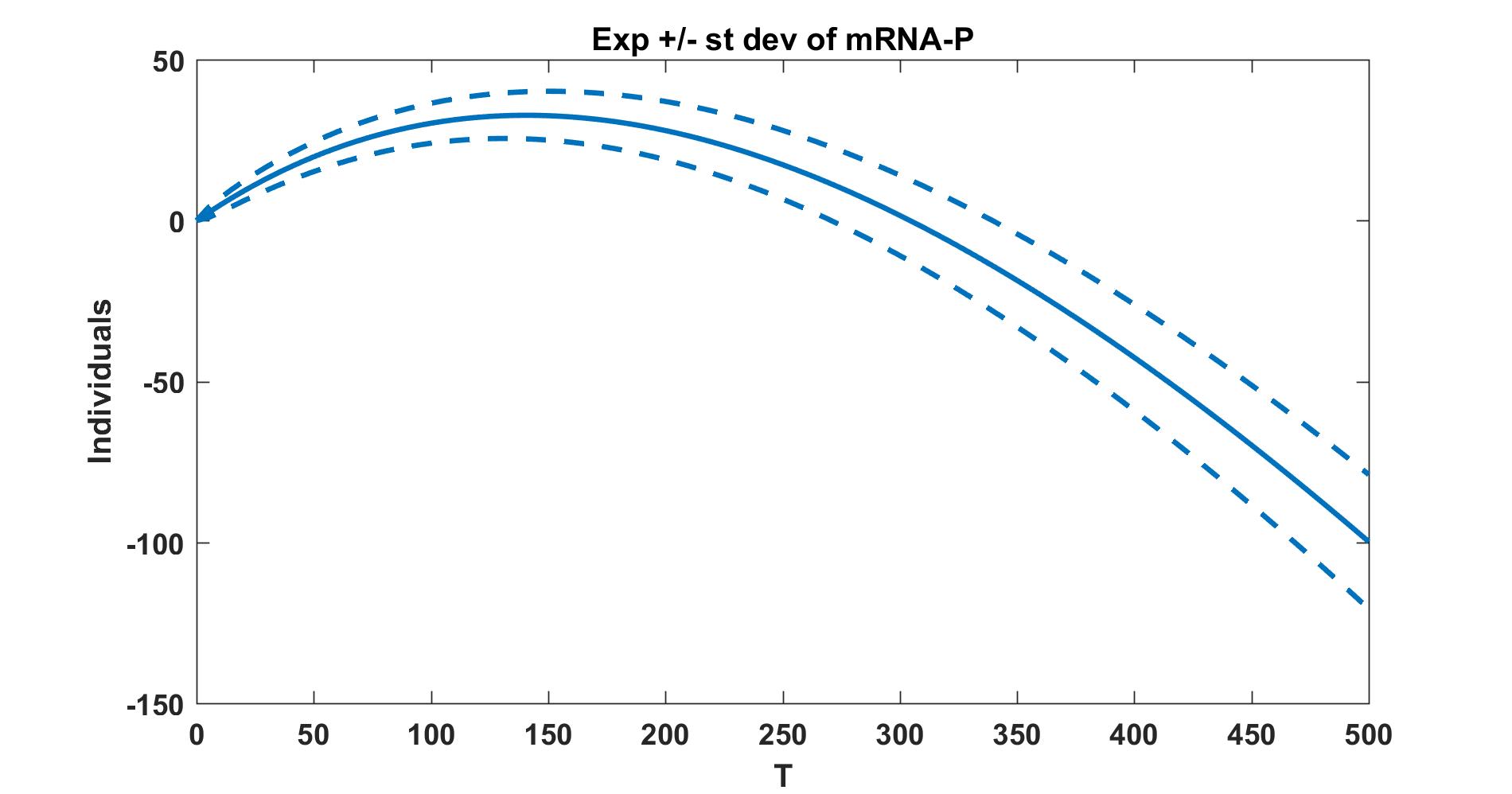}
\caption{$N \rho_{I}(\hat X^N,T),T \in [0,100]$ for reward structure  $\rho^{mRNA-P}$. }
\label{FigureInstantaneousReward}
\end{figure}%
Note that, in this case, the resulting variance and expectation, as estimated by the CLA, are guaranteed to be exact for any $N$. This is because the SRN is linear \cite{grima2015linear}. It is easy to observe that our algorithms are exponentially faster than the computation of the same measures  on the original CTMC, because of the continuous nature of the CLA.
\end{example}

\subsubsection{Cumulative Rewards}
Cumulative rewards can also be computed exploiting $\hat Y^N,$ the CLA approximation of $\hat X^N$, as shown in the following proposition
\begin{proposition}
Let $T\in\mathbb{R}_{\geq 0}.$ Then, $\rho^{CLA}_C(\hat x,\hat Y^N,T)$, the cumulative reward for $\hat Y^N$ starting from $\hat Y^N=\hat x$, can be computed as follows 
\begin{align*}
\rho^{CLA}_C(\hat x,\hat Y^N,T)=\int_{0}^T\rho_I^{CLA}(\hat x,\hat Y^N,s) ds,
\end{align*}
\end{proposition}
\begin{proof}
Let $\omega:\mathbb{R}_{\geq 0}\to \mathbb{R}^{|\Lambda|}$ be a path of $\hat Y^N$. Then, we have that
$$\rho^{CLA}_C(\hat x,\hat Y^N,T)=\mathbb{E}[\rho_C(\omega,T)|\omega(0)=\hat x]=\mathbb{E}[\int_{0}^T \rho(\omega(t))dt\,|\, \omega(0)=\hat x].$$
Now, in order to conclude, we can apply Fubini's theorem \cite{schilling2017measures}, which allows us to compute a double integral using iterated integrals. Thus, switching the order of integration. Being both a probability measure and the Lebesque measure over the reals $\sigma-$finite measures, a sufficient condition for application of Fubini's theorem is that
$\mathbb{E}[\int_{0}^T |\rho(\omega(t))|dt]$ is finite. Owing to boundedness of $\rho$, there is an $M>0$ such that, 
for all $x \in \mathbb{R}^{|\Lambda|}$, we have that $|\rho(x)|\leq M$. Thus, we have,
$$\mathbb{E}\left[\int_{0}^T |\rho(\omega(t))|dt\right]\leq \mathbb{E}\left[\int_0^T M dt\right]= MT,$$
which is finite as $T$ and $M$ are both finite.
\end{proof}

The following proposition, for $\hat x\in \mathbb{R}^{|\Lambda|}$, guarantees that $\rho^{CLA}_C(\hat x,\hat Y^N,T)$ converges to $\rho_C(\hat x, \hat X^N,T)$
\begin{proposition}
\label{PropositionCumulativeRewards}
Let $ T \in \mathbb{R}_{\geq 0}$, then it holds that
$$ \lim_{N\to \infty} \rho_C(\hat x, \hat X^N,T)=\lim_{N\to \infty} \rho^{CLA}_C(\hat x,\hat Y^N,T)$$
\end{proposition}
\begin{proof}
For a path $\omega:\mathbb{R}_{\geq 0}\to \mathbb{R}^{|\Lambda|}$, define the following functional $\mathcal{R}_C (T,\omega)=\int_0^T \rho(\omega(t))dt,$ which is defined on $\mathcal{D},$ the space of $\mathbb{R}^{|\Lambda|}-$valued Cadlag functions (i.e. right continuous functions with left limit) \cite{billingsley2013convergence}. 
$\rho_C(\hat x, \hat X^N,T)=\mathbb{E}[\mathcal{R}_C (T,\omega)]$, where the expectation is taken over $\Omega_{\hat x}$, the paths of $\hat X^N$ starting from $\hat x$. 
As $T$ and $\rho$ are bounded, for each  $\omega$, $\mathcal{R}_C (T,\omega)$ is bounded. It is also continuous, due to the continuity of $\rho$.  Thus, we can apply same reasoning as in the proof of Proposition \ref{ProofInstantaeosReward},  applying Portmanteau theorem to conclude.
\end{proof}

\begin{example}
We consider the SRN introduced in Example \ref{RunningExample}. We are interested in knowing the expected cumulative reward of $mRNA-P$ to understand if during the time interval there have been, on average, more $mRNA$ or more $P$ molecules in the system. This can be computed using the reward structure $\rho^{mRNA-P}$ introduced in Example \ref{Example-InstanRew}, and the following cumulative reward: $$N R_{=?}[C^{\leq T}_{\rho^{mRNA-P}}],\,T \in [0,500],$$
where $R_{=?}[C^{\leq T}_{\rho^{mRNA-P}}]$ is intended to be computed on $\hat X^N.$ 
The resulting expectation and variance are plotted in Figure \ref{FigureCumulativeReward}.
\begin{figure}
	\centering
     \includegraphics[width=1\linewidth]{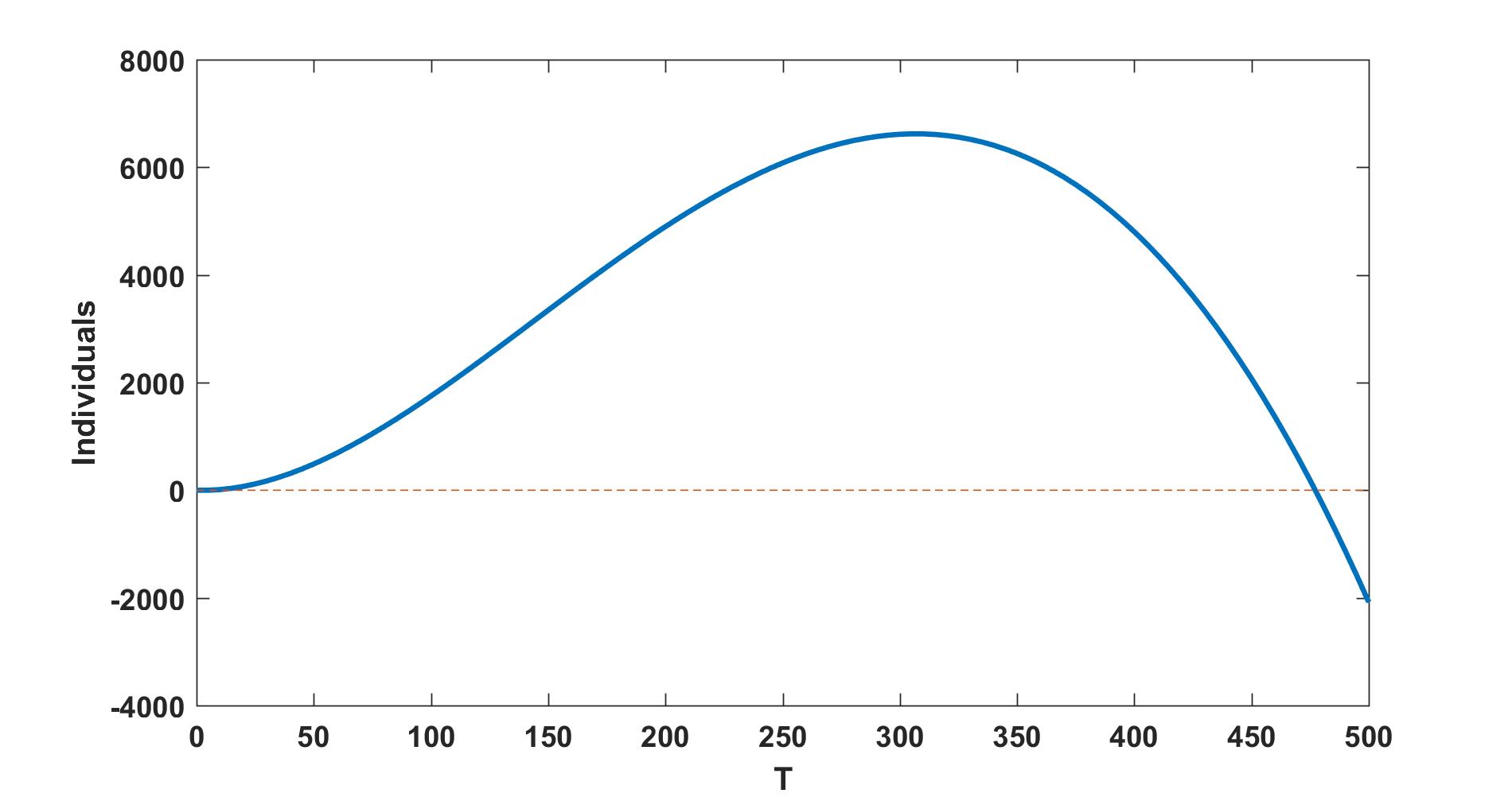}
\caption{$N\rho_{C}(\hat X^N,T),T \in [0,500],$ for reward structure  $\rho^{mRNA-P}$. }
\label{FigureCumulativeReward}
\end{figure}%
We stress again how in this case, since the SRN is linear, the measure estimated by the CLA is exact for any $N.$
\end{example}

\subsubsection{Bounded Reachability Rewards}
Bounded reachability rewards can be computed efficiently on the CLA under a further assumption on the reward structure $\rho$. Specifically, for $\hat x \in \mathbb{R}^{|\Lambda|}_{\geq 0}$, consider the predicate $\eta(\hat x)=B \hat x<b,$  $b \in (\mathbb{R}\cup\{\infty\})^m$, $m>0$. We assume that the reward structure is defined on the projection of $\hat X^N$ spanned by the colums of matrix defining the $\eta$ predicate, namely $\rho: \mathbb{R}^m \to \mathbb{R}$ assigns a reward to each state of $B \hat X^N$.
Consider the CSL reward property $R_{\sim r}[F_{\rho}^{\leq T} \eta]$. That is, $R_{\sim r}[F_{\rho}^{\leq T} \eta]$ is the bounded reachability reward until reaching a state in $A=\{\hat x\in \mathbb{R}^{|\Lambda|} \,s.t.\, \forall i \in \{1,..,m \}, (B \hat x)_i\geq b_i \}$ during the time interval $[0,T]$. Such a reward can be computed by exploring the approximation of the CLA in terms of the DTMC $\hat Z^{\Delta z,h,N},$ which is obtained by time and space discretization of the process $\hat Z^N=B \hat X^N$.
We call $\rho_{reach}(\hat x_0, \hat Z^{\Delta z,h,N},\lfloor \frac{T}{h}  \rfloor,A)$ the bounded reachability reward computed on $\hat Z^{\Delta z,h,N}$ for a number of discrete steps $\lfloor \frac{T}{h}  \rfloor,$ where $h>0$. Then % and with an abuse of notation we call $A$ both the target set in the original Markov chain, and the correspondent target set in the discretized DTMC, $\hat \hat Z^{\Delta z,h,N}$.
$\rho_{reach}(\hat x_0, \hat Z^{\Delta z,h,N},\lfloor \frac{T}{h}  \rfloor,A)$ can be computed by considering the modified process $\hat Z^{\Delta z,h,N}$ where the target states are made absorbing, and modifying the reward structure $\rho$ to $\bar{\rho}$ such that $\bar{\rho}(\hat x)=0$ for all absorbing states. Then, for $\hat x_0 \in \mathbb{R}^{|\Lambda|}$ and $z_{d,0},$ the state in the state space if $\hat Z^{\Delta z,h,N}$ corresponding to the region containing $\hat x_0,$ cumulative rewards can be computed with the following algorithm for $n>0$:
\begin{align}\nonumber
\rho_{reach}(&\hat x_0, \hat Z^{\Delta z,h,N},n,A)=\\
&\sum_{ z \in S^{\Delta z}}\bar{\rho}(z)Prob(\hat Z^{\Delta z,h,N}(n-1)=z| \hat Z^{\Delta z,h,N}(0)=z_{d,0})h\nonumber \\
&+ \rho_{reach}(\hat Z^{\Delta z,h,N},\hat x_0,n-1,A).\label{EqnBoundedReach}
\end{align}
and such that $\rho_{reach}(\hat Z^{\Delta z,h,N},\hat x_0,0,A)=0$ for $\hat x_0 \not \in A$.
The proof of the following proposition can be found in the Appendix. 
\begin{proposition}
\label{PropositionBoundedRewa} 
For $T\in \mathbb{R}_{\geq 0}$ and $B\in\mathbb{R}^{|\Lambda|\times k}$ let $A$ be the set defined as $A=\{x\in \mathbb{R}^{|\Lambda|} \,s.t.\, \forall i \in \{1,..,k \}, (Bx)_i\geq b_i \}$. Then, for $\hat x_0 \in \mathbb{R}^{|\Lambda|}$ and $z_{d,0}$, the state in the state space of $\hat Z^{\Delta z,h,N}$ corresponding to the region containing $\hat x_0,$ it holds that:
$$ \lim_{N\to \infty}\lim_{h\to 0}\lim_{\Delta z \to 0}|\rho_{reach}(\hat x_0,\hat X^N,T,A)- \rho_{reach}(z_{d,0}, \hat Z^{\Delta z,h,N},\lfloor \frac{T}{h}  \rfloor,A)|=0.$$
\end{proposition}

\begin{example}\label{example:reachabilityRewards}
We consider the SRN introduced in Example \ref{RunningExample}. We are interested in knowing the expected cumulative reward of $mRNA-P$ for all those paths for which the $mRNA$ does not reach $30$ individuals within $[0,T]$ for $T\in [0,50]$.
We consider the reward structure $\rho^{mRNA-P}(x)$ introduced in Example \ref{Example-InstanRew}, and the following cumulative reward $\rho_{reach}( X^N,T,mRNA<30)=N \rho_{reach}( \hat X^N,T,mRNA<\frac{30}{N}),T \in [0,50].$
To compute such a reward we explore the CLA approximation of $X^N$. We consider $B_1=[1, 0],$ $B_2=[-1,1]$ and $B=[B_1,B_2],$ where we assume the first component of $X^N$ represents the number of protein molecules in that state. Then, we consider $\hat Z^N$, the projection of the CLA of $\hat X^N$ over B, namely, $\hat Z^N=B \hat Y^N$. At this point we discretize $\hat Z^N$ with sampling time $h>0$ and a grid of cells of size $dz>0$, and compute the above rewards using Eqn \eqref{EqnBoundedReach}.
\begin{figure}
	\centering
     \includegraphics[width=0.85\linewidth]{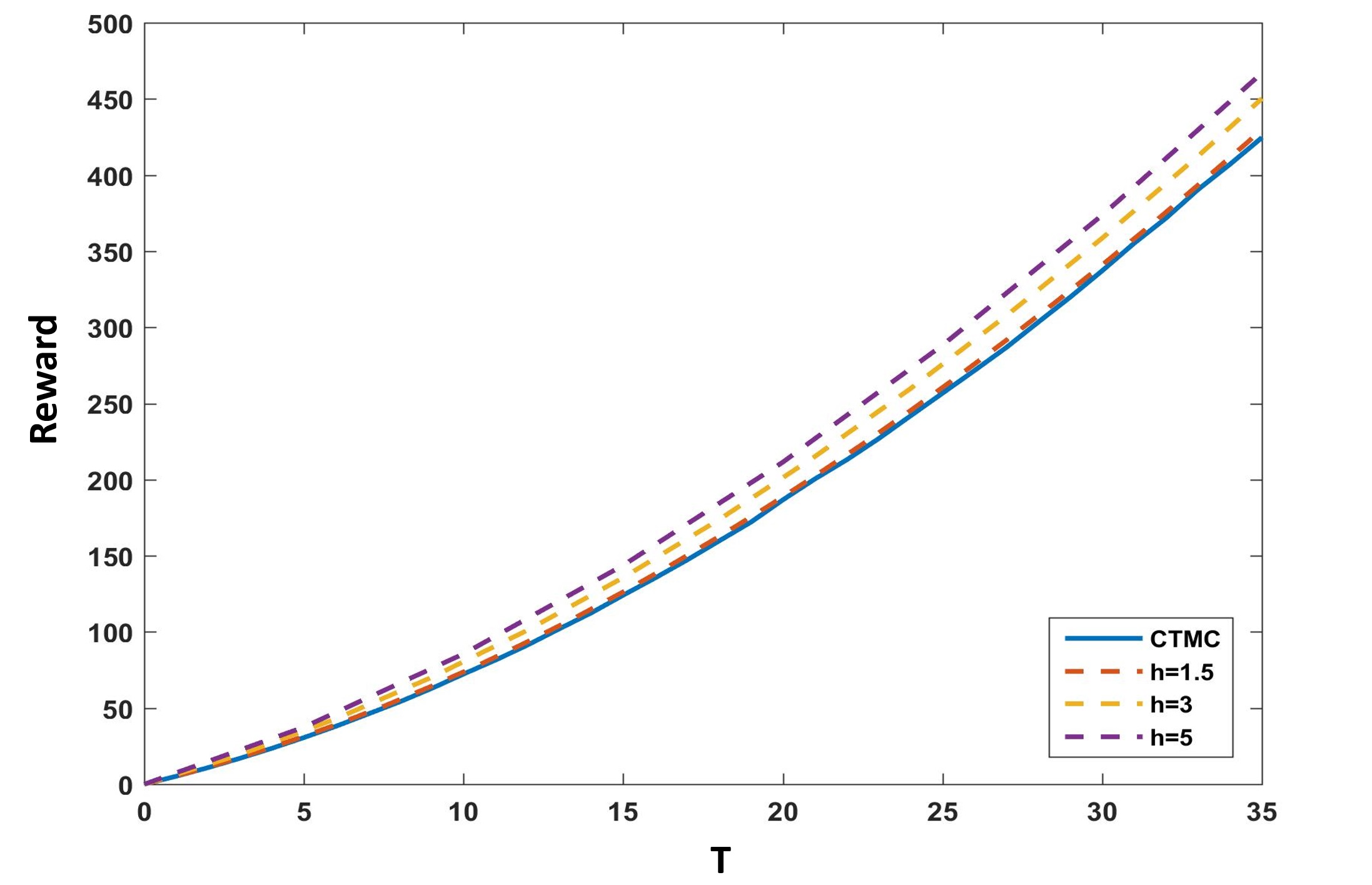}
\caption{$\rho_{reach}(x_0,T),T \in [0,35],$ for reward structure  $\rho^{mRNA-P}$ estimated using $10000$ simulations compared with the CLA approximation for different values of sampling time $h$. $dz=0.5$ for all experiments. }
\label{FigureCumulativeBOundedReward}
\end{figure}%
The solution to Eqn \eqref{EqnBoundedReach} is approximate, and errors are introduced by two factors: firstly, by the usage of the CLA approximation of $X^N$, and, secondly, by the discretization of the resulting Gaussian process. Thus, we compare our reward value with the value computed on $X^N$ using $10000$ simulations for each time point. The resulting values are plotted in Figure \ref{FigureCumulativeBOundedReward}. To perform a further comparison we employ the following metrics, $\epsilon_{max}^{rel}$ and $\epsilon_{avg}^{rel}$, defined as follows: $$\epsilon_{max}^{rel}=max_{T\in \Sigma_h}\frac{|Rew^{CLA}_T-Rew_T)|}{|Rew|},\quad \epsilon_{avg}^{rel}=\sum_{T \in \Sigma_h}\frac{|Rew^{CLA}_T-Rew_T)|}{|Rew|}\frac{1}{|\Sigma_h|}$$
where $\Sigma_h$ is the set of sampling points for sampling time $h,$ $Rew^{CLA}_T=\rho_{reach}^{mRNA-P}$ $(\hat Z^{\Delta z,h,N}, \hat x_0,\lfloor \frac{T}{h}  \rfloor,A)$ and $Rew_T=\rho_{reach}^{mRNA-P}(\hat X, \hat x_0,{T} ,A)$. The calculated metrics are summarised in the table below for three different values of $h$.
\begin{center}
	\begin{tabular}{|l|l|l|}
	\hline
	h            &  \,	 	    $\epsilon_{avg}^{rel}$ & \,$\epsilon_{max}^{rel}$ \, \\ 
	\hline
$5$       	 & 	 $1.5468$  &  $7.96$  \\
$3$        &  $0.196$	  & $0.88$   \\
$1.5$       &  	  $0.041$	  & $0.24$ \\
	\hline
	\end{tabular}
\end{center}
It is possible to observe how the two measures converge very fast. In fact, already for $h=1.5$, which corresponds to $25$ sampling times, the two measures have an average relative error of less than $0.041.$

\end{example}
%\begin{example}
%We consider the RN introduced in Example \ref{RunningExample}. We are interested in knowing the expected number of protein for all those paths for which the $mRNA$ does not reaches a value of $30$ individuals within $T<100$. 
%We consider $B=[1,0], B'=[0,1]$, and $\bar{B}=[B; B']$, where we assume an ordering for the species such that $X^N=[X^N(mRNA),X^N(Pro)]$.
%We consider the reward structure $\rho_{Pro}(x)=x(Pro),$ where $x(Pro)$ is the number of protein in state $x$, and the following cumulative reward $\rho_{F}(x_0,[30,\infty],T),T \in [0,100].$
%Then, we compute the above rewards using Eqn \eqref{BoundedReach}.
%{\color{red}This reward can be implemented in PRISM by putting protein to $0$ for all paths for which $mRNA\geq 30$}
%The resulting expectation and variance is plotted in Figure
%\begin{figure}
%	\centering
%     \includegraphics[width=1\linewidth]{CumulativeReward.jpg}
%\caption{$\rho_{F}(x_0,[30,\infty],T),T \in [0,100]$ for reward structure  $\rho_{mRNA}$. }
%\label{FigureInstantaneousReward}
%\end{figure}%
%THe resulting variance and expectation, as extimated by the LNA are exact because the RN is linear. However, it is exponentially faster than the computation on the original CTMC, because of the continuous nature of the LNA.

%\end{example}

\section{Experimental Results}\label{Exp}
We implemented our algorithms in Matlab and evaluated them on two case studies. All the experiments were run on an Intel Dual Core i$7$ machine with $8$ GB of RAM. The first case study is a Gene Expression Network as introduced in Example \ref{RunningExample}. We use this example to demonstrate that our approach is more powerful than existing approximate techniques. Specifically, we show how our CLA approach, based on a Gaussian process approximation, is able to correctly evaluate properties that methods based on Fluid Limit Approximation \cite{bortolussi2012fluid} cannot, while still guaranteeing comparable scalability.  The second example is a  Phospohorelay Network with $7$ species. We use this example to show the trade-off between the different parameters and the molecular population. More precisely, we show that the accuracy of our approach increases as the number of molecules grows, but can still give fast and accurate results when the molecular population is relatively small. We validate our results by comparing our method with statistical model checking (SMC) as implemented in {PRISM} \cite{kwiatkowska2011prism}. In fact, for both examples, exact numerical computation of the reachability probabilities using uniformisation techniques on the induced CTMC is not feasible because of state space explosion.

\subsection{Gene Expression}
%\todo{You should avoid putting words in between dollar signs, becasue Latex typesets them in a strange way - instead use in math mode \mbox{\emph{Time}}}
We consider the following gene expression model, as introduced in Example \ref{RunningExample}
with initial counts of all the species equal to $0$. 
We consider the property $P_{=?}(F^{[0,Time] }(mRNA \geq 175)$,
which quantifies the probability that at least $175$ $mRNA$ molecules are produced during the first $Time$ seconds, for $Time\in[0,1000]$.
This is a particularly difficult property because the trajectory of the mean-field of the model, and so the expected value of the CLA, does not enter the target region. As a consequence, approximate approaches introduced in \cite{ethier2009markov} and \cite{bortolussi2014stochastic}, which are based on the hitting times of the mean-field model, fail and evaluate the probability as always equal to $0$.

\begin{figure}
  \centering
\includegraphics[width=0.72\linewidth]{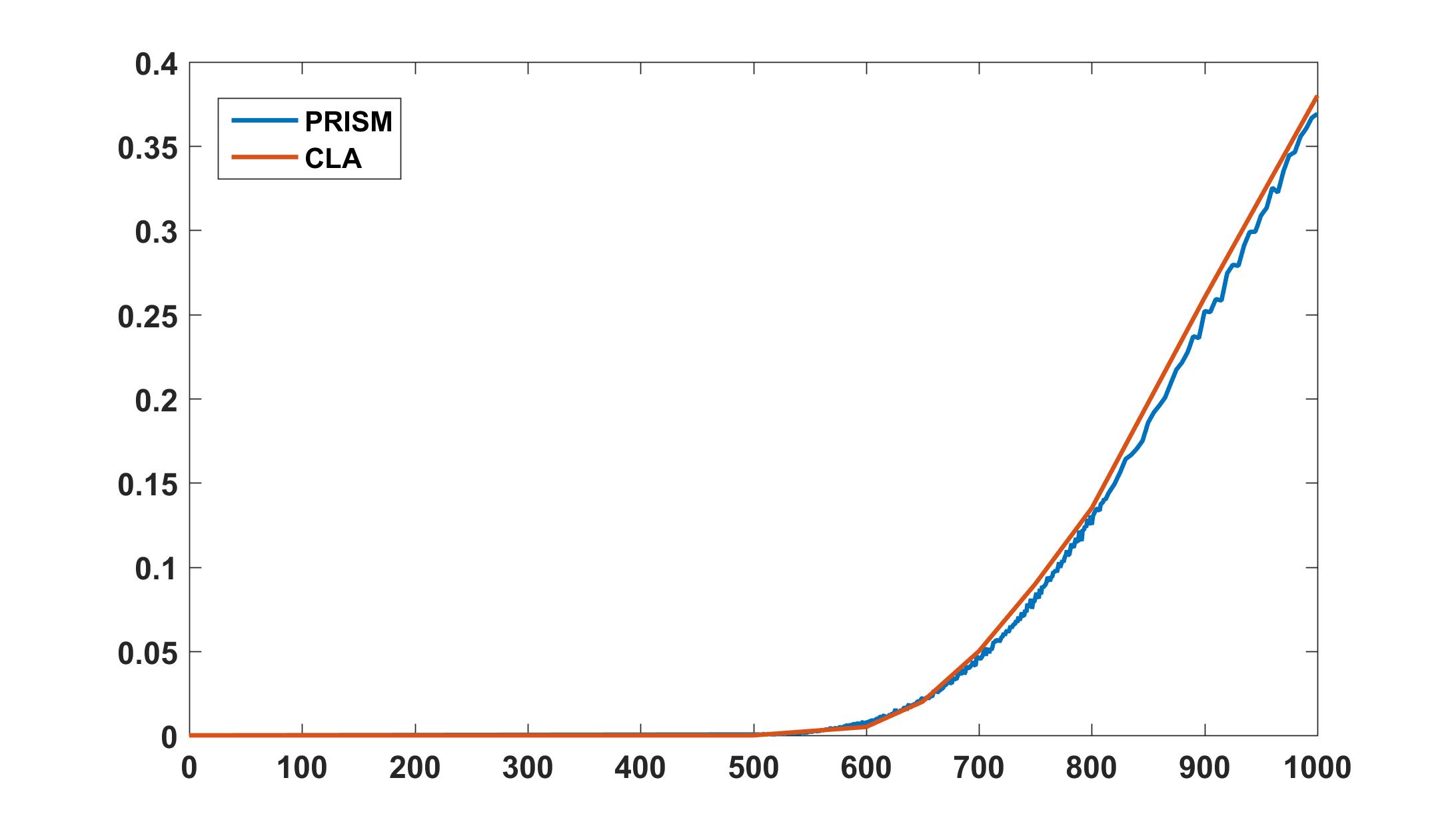}
  \caption{The figure plots $F_{=?}[mRNA\geq 174]_{[0,Time]}$ for $h=1.85$ and $\Delta z=0.5$. The x-axis represents the value of $Time$ and the y-axis the quantitative value of the formula for that value of $Time$.}
    \label{fig:Gene}
\end{figure}
Conversely, our approach is able to correctly evaluate such a property. Figure \ref{fig:Gene} compares the value computed by our method with statistical model checking of the same property as implemented in PRISM over $30000$ simulations for each time point and confidence interval $0.01$.
In Figure \ref{fig:Gene} we consider $h=1.8$ and $\Delta z=0.5$ and demonstrate that
our approach is able to correctly estimate such a difficult property.
%Results show that our approach is able to correctly estimate such a difficult property. 
Note that, as the mean-field does not enter the target region, for each time point the probability to enter the target region depends on a portion of the tail of the Gaussian given by the CLA. 
As a consequence, the accuracy of our results strictly depends on how well the CLA approximates the original CTMC, much more than for properties where the mean-field enters the target region. In the following table, we evaluate our results for two different values of $h$ and $\Delta z=0.5$. In order to compare the accuracy we consider the following metrics as defined in Example \ref{example:reachabilityRewards}, $\epsilon_{avg}^{rel}$ and $\epsilon_{max}^{rel}$. 
%\todo{But you had these metrics in Example 7.7 without the norm definition}
The comparison is shown in the table below.
%As a consequence, as shown in the following table for $\Delta z=0.5$, %more than for other properties, 
%\todo{Marta: removed ore than for other properties, not sure what was meant.
%the accuracy of the results depends on the accuracy of the LNA and the choice of $h$.
%for a good accuracy is crucial the choice of $h$
\begin{center}
	\begin{tabular}{|l|l|l|l|l|}
	\hline
	Property            &  \,	Ex.\,	Time	 \,  & \, h  \, 	  &	\, 	    $\epsilon_{avg}^{rel}$ & \,$\epsilon_{avg}^{max}$ \, \\ 
	\hline
$F_{=?}[mRNA\geq 174]_{[0,Time]}$, $Time\in [0,100]$       & 	$298$	   sec           	 & 	$1.85$ &  	$0.0075$  & $0.022$  \\
$F_{=?}[mRNA\geq 174]_{[0,Time]}$, $Time\in [0,100]$       &  $152$		   sec           	 & 	$5$ &  	$0.0147$  & $0.13$  \\
	\hline
	\end{tabular}
\end{center}

\subsection{Phosphorelay Network}
We now consider a three-layer phosphorelay network consisting of 7 species given by the following reactions:
\begin{align*}
    & L1 + B \to^{\frac{0.01}{N}\cdot L1 \cdot B} B + L1p;     \quad L2 + L1p \to^{\frac{0.01}{N} \cdot L2 \cdot L1p} L1 + L2p;\\
    &\quad  L3  + L2p \to^{\frac{0.01}{N}\cdot L3 \cdot L2p} L3p + L2;\quad  L3p \to^{0.1 \cdot L3p} L3;\\
    &\quad L2p  \to^{0.01 \cdot L2p} L2 ; \quad L1p \to^{0.01 \cdot L1p} L1 .
\end{align*}
There are $3$ layers, $(L1,L2,L3)$, which can be found in phosphorylate form $(L1p,L2p,L3p)$, and the ligand $B$.
We consider the initial condition $x_0\in \mathbb{N}^{7}$ such that $ x_0(L1)=x_0(L2)=x_0(L3)=0.25 N $, $ x_0(L1p)=x_0(L2p)=x_0(L3p)=0$ and $x_0(B)=0.15 N$. In Figure \ref{Phopho}, we compare the estimates obtained by our approach for two different initial conditions ($N=400$ and $N=800$) with statistical model checking as implemented in PRISM \cite{kwiatkowska2011prism}, with $30000$ simulations and confidence interval equal to $0.01$. In both experiments we set $\Delta z=0.5$.
%\todo{I would remove Figure 7. We already have the table. It does not help (and it is not well looking)}
%\todo{But the figure is more informative}
\begin{figure}
	\centering
\begin{subfigure}{.5\textwidth}
     \centering
     \includegraphics[width=1.\linewidth]{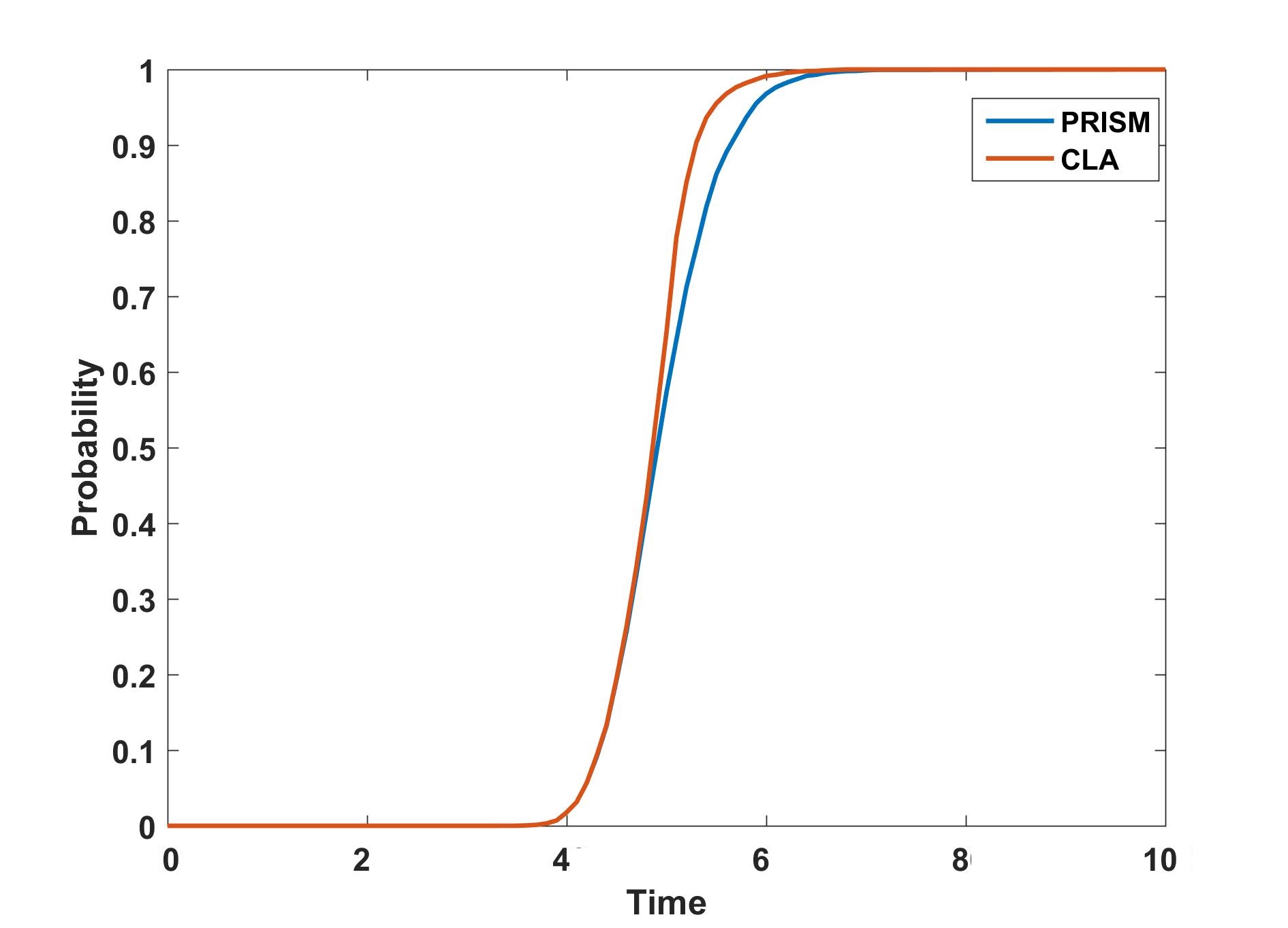}
     \caption{}
     \label{fig:100}
   \end{subfigure}%
\begin{subfigure}{.5\textwidth}
  \centering
  \includegraphics[width=1\linewidth]{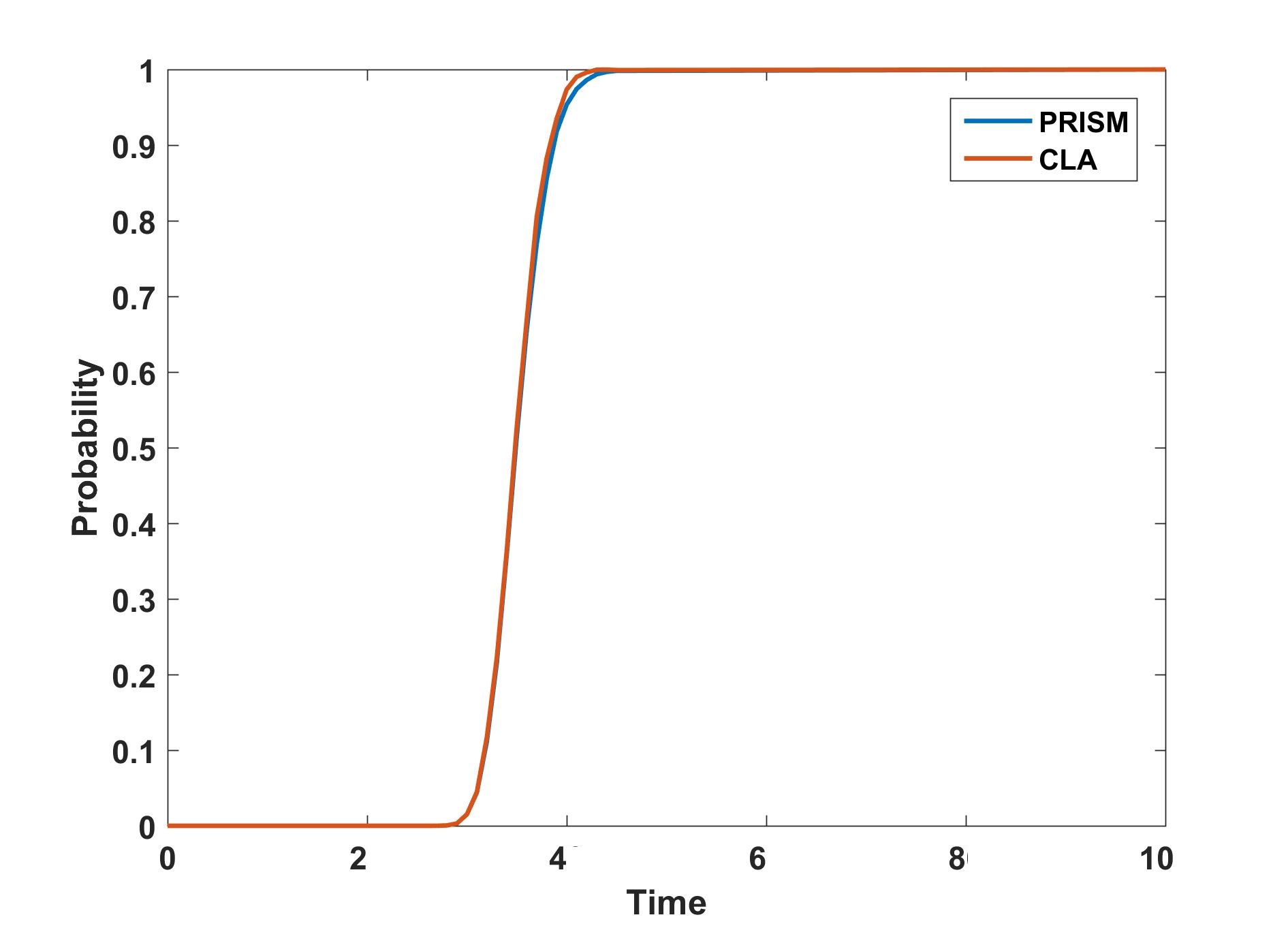}
  \caption{}
  \label{fig:180}
\end{subfigure}
\caption{Comparison of the evaluation of $F_{[0,Time]}[L3p>80]$ (a) %(Fig \ref{fig:100}) 
with $N=400$ and $F_{[0,Time]}[L3p>180]$ (b) %(Fig \ref{fig:180}) 
with $N=800$ using statistical model checking as implemented in PRISM and our approach. In both figures we considered $h=0.1$, $\Delta z=0.5$. }
\label{Phopho}
\end{figure}%

In Figure \ref{fig:100} we can see that, if we increase the time interval of interest, the error tends to increase. This is because, for $N=400$, the CLA and CME do not have perfect convergence. As a consequence, at every step of the discretized DTMC, a small error is introduced. This source of error is present until we enter the target region with probability $1$. If we increase $N$ the error disappears, and the inaccuracies are due to the finiteness of $h$ and $\Delta z$. However, already for $h=0.1$ and $N=800$, the CLA produces a fast and reasonably accurate approximation. 
In the following table we compare our approach and PRISM evaluations for different values of $N$ and $h$ and $\Delta z=\frac{0.5}{N}$ in the normalised space, which implies the resulting discrete state process takes values in $\mathbb{Z}$.
\begin{center}
	\begin{tabular}{|l|l|l|l|l|l|}
	\hline
	Property            &  \,		Ex.\,Time	 \,  & \, h  \, 	  & \, N \, &	\, 	    $\epsilon_{avg}^{rel}$ & \,$\epsilon_{avg}^{max}$ \, \\ 
	\hline
$F_{=?}[L3p>80]_{[0,Time]}$, $Time\in [0,10]$       & 	$97$	   sec           	 & 	$0.1$ & $400$ & 	$0.0088$  & $0.11$  \\ 
	$F_{=?}[L3p>180]_{[0,Time]}$, $Time\in [0,10]$    & 	$130$	     sec      	 & 	 	$0.1$     &  $800$ & $0.0015$  & $0.0217$\\ 
	$F_{=?}[L3p>80]_{[0,Time]}$, $Time\in [0,10]$    & 	$28$	     sec      	 & 	 	$0.5$     &  $400$ & $0.0381$& $0.24$  \\  
		$F_{=?}[L3p>180]_{[0,Time]}$, $Time\in [0,10]$    & 	$39$	    sec      	 & 	 	$0.5$     &  $800$ & $0.0289$ & $0.14$  \\  
	\hline
	\end{tabular}
\end{center}
The results show that the best accuracy is obtained for $h=0.1$ and $N=800$, where $h=0.1$ induces a finer time discretization, whereas the worst are for $h=0.5$ and  $N=400$. 
We comment that the numerical solution of the CME using PRISM is not feasible for this model, and our method is several orders of magnitude faster than statistical model checking with PRISM ($30000$ simulations for each time point), which takes several hours for each property.

\section{Conclusion}
We presented a framework for approximate model checking of a time-bounded fragment of CSL extended with rewards for CTMCs that are induced from Stochastic Reaction Networks.
Our approach employs an abstraction based on the Central Limit Approximation to approximate the CTMC as a Gaussian process. Then, numerical procedures for model checking CSL formulae on the resulting Gaussian process are derived. 
We do not consider time-unbounded properties because of the nature of the convergence of CLA, which is guaranteed just for finite time.
Since the CLA requires solving a number of differential equations that is quadratic in the number of species and independent of the population size, our methods enable formal analysis of possibly infinite-state CTMCs that cannot be analysed using classical methods based on uniformization \cite{DHK14,wolf2010solving}.  
%We demonstrate that our approach does not suffer from the state space explosion problem, thus enabling formal analysis of CTMCs that cannot be analysed using classical methods based on uniformization and with infinite state space \cite{DHK14,wolf2010solving}.  

Deriving model checking algorithms was challenging because the CLA yields a continuous time stochastic process with an uncountable state space.
As a consequence, existing methods that rely on finite state spaces cannot be used directly and discretizing the uncountable state space induced by the CLA naturally leads to state space explosion. In order to overcome these issues, we considered reachability regions defined as  polytopes. Using the fact that the CLA is a Gaussian Markov process, for a given linear combination of the species of a SRN we are able to project the original, multi-dimensional Gaussian process onto a uni-dimensional stochastic process. We then derived an abstraction in terms of a time-inhomogeneous DTMC, whose state space is independent of the number of the species of a SRN, as it is derived by discretizing linear combinations of the species. This ensures scalability.
On different case studies, we showed that our approach outperforms existing methods and permits fast and scalable probabilistic analysis of SRNs. 
The accuracy depends on parameters controlling space and time discretization, as well as on the accuracy of the CLA. Using the theory of convergence in distribution we proved the convergence of our algorithms in the limit of high populations. 
As a future work we plan to release a tool for scalable model checking of SRNs. Moreover, we wish to investigate the speed of convergence of our methods.  
%Moreover, since the LNA is accurate in the limit of the molecular population, we can increase the accuracy of our estimates by increasing the molecular populations.

\appendix
\section{Proofs}

\noindent\textbf{Theorem \ref{COnvergenceReach}} Let $\mathcal{C}=(\Lambda,R)$ be a SRN with induced CTMC $\hat X^N$ and $\hat Z^{\Delta z,h,N}$ be the DTMC obtained by space and time discretization of $B \hat Y^N$. Assume $\hat X^N(0)=\hat x_0$ and the corresponding initial state for $\hat Z^{\Delta z,h,N}$ is $z_{d,0}$.
Then, for $t_1,t_2 \in \mathbb{R}_{\geq 0}$, and $A=\{x \in \mathbb{R}^{|\Lambda|}_{\geq 0}\, s.t.\,\forall i \in \{1,...,m \} (Bx)_i\leq b_i \},$ for $B\in \mathbb{R}^{m\times |\Lambda|}$ and $b \in \mathbb{R}^m$, it holds that:
\begin{align*}
     \lim_{N\to \infty}\lim_{h \to 0}\lim_{\Delta z \to 0}| P_{reach}^A(\hat x_0,t_1,t_2)\, - \, 
     &P_{reach}^{\Delta z,h,A}(z_{d,0},t_1,t_2)|  =0.
\end{align*}

\noindent\textbf{Proof.}
Without any loss of generality, we assume $t_1=0,t_2=T$.
Call $$P_{reach}^{h,A}(\hat x_0,0,T)=Prob^h(\exists t \in [0,\lceil\frac{T}{h}\rceil ] \, s.t.\,  \hat {Z}^{h,N}(k) \leq b\,|\,\hat {Z}^{h,N}(0)=B\hat x_0 ),$$ 
and 
$$  P_{reach}^{\hat Y^N,A}(\hat x_0,0,T)=Prob^{\hat Y^N}(\exists t \in [0,T ] \, s.t.\, \hat Y^N(t) \in A \,|\,\hat Y^N(0)=\hat x_0 ), $$
where $Prob^{Y^N}$ is the Gaussian probability measure of $\hat Y^N$.

By application of the triangular inequality we have that:

\begin{align*}
|P_{reach}^A(\hat x_0,0,T)  -&P_{reach}^{\Delta z,h,A}(z_{d,0},0,T) |\leq \\
&|P_{reach}^A(\hat x_0,0,T)  -P_{reach}^{\hat Y^N,A}(\hat x_0,0,T) | +\\
&|P_{reach}^{\hat Y^N,A}(\hat x_0,0,T)  -P_{reach}^{h,A}(\hat x_0,0,T) | + \\
&|P_{reach}^{h,A}(\hat x_0,0,T)  -P_{reach}^{\Delta z,h,A}(z_{d,0},0,T) |.
\end{align*}

The convergence of the third and second components is a consequence of Theorems \ref{Space Discretization} and \ref{Time Discretization}. We need to show that: 

$$ \lim_{N \to \infty}|P_{reach}^A(\hat x_0,0,T)  -P_{reach}^{\hat Y^N,A}(\hat x_0,0,T) |=0 .$$

Note that we removed the limits for $\Delta z$ and $h$, as this term is independent of time and space discretization. 
In what follows we assume that $BX^N$ is a uni-dimensional process. Generalization for $m>1$ follows from this case. 
Intuitively, this holds due to the convergence of $X^N$ to its CLA $Y^N$. A formal proof requires a more involved machinery. In fact, Theorem~\ref{lnconv} states that: 

\[ \sqrt{N}\left(\hat X^N(t) - \Phi(t) \right) \Rightarrow G(t), \]
hence, to rely on it, we need to reason on the modified stochastic model: \[G^N(t) = \sqrt{N}\left({\hat X_N(t)} - \Phi(t) \right),\] rather than on the original CTMC $\hat X^N(t)$. Now, consider the reachability problem $B \hat X^N \leq b$; rephrasing it in terms of $G^N$ we get:

\[B {\hat X^N(t)} \leq b\ \text{iff}\ B G^N(t) \leq \sqrt{N}(b - B\Phi(t)) = b^N(t). \]

As we can see, the reachability problem for $G^N$ has a different nature: the threshold $b$ becomes both $N$ dependent and time dependent! In addition, we see that for the CLA, $B Y^N(t) \leq b$ iff $B G(t) \leq b^N(t)$. Let's look at this reachability problem from the point of view of the trajectory space, i.e. the space of cadlag function  $\omega : \mathbb{R}_{\geq 0}\to \mathbb{R}$. Both $G^N$ and $G$ can be seen as probability measures over this space.  The reachable set in the trajectory space depends on $N$, precisely being $R_N = \{\omega~|~\exists t \in [0,T]: \omega(t) \leq b^N(t) \}$. We also consider the complement of this set, $R_N^c = \{\omega~|~\forall t \in [0,T]: \omega(t) > b^N(t) \}$, and the boundary of the set $\partial R_N = \{\omega~|~\forall t \in [0,T]: \omega(t) \geq b^N(t) \wedge \exists t \in [0,T]: \omega(t) = b^N(t) \}$.

Before proceeding further, we need to understand how the set $R_N$ changes as $N$ goes to infinity. Consider the threshold $b^N(t) = \sqrt{N}(b - B\Phi(t))$. There are three cases: 

\begin{enumerate}

    \item if $b > B\Phi(t)$, then $b^N(t) \rightarrow +\infty$;

    \item if $b < B\Phi(t)$, then $b^N(t) \rightarrow -\infty$;

    \item if $b = B\Phi(t)$, then $b^N(t)=0$. 

\end{enumerate}
In the first case, the reachable set at time $t$ converges to $\mathbb{R}$, in the second case to the empty set, in the third case to $(-\infty,0]$. Therefore, the limit reachable set $R$ in the trajectory space will be the union for each $t$ of one of these three kind of sets.

%For simplicity, we make the following assumption: the rate functions of the SRN are real analytic functions (cite Fluid MC). This is not restrictive: all rate functions used in systems biology are in fact real analytic (this includes polynomials, Hill and Michaelis Menten). 

By the assumption that rate functions are real analytic, it follows that $\Phi(t)$ is also a real analytic function, and therefore $B \Phi(t)$ will equal $b$ only in a finite number of points of $[0,T]$, or in the whole interval (a degenerate case which is easily tractable) \cite{krantz_realanalytic_2002}. It then follows that $b(t) = \lim_{N\rightarrow\infty} b^N(t)$ changes value a finite number of times, say at times $t_1, \ldots, t_n$, where it equals zero. Outside these points, it is either plus or minus infinity. The reachable set $R$, in the limit of infinite $N$, is thus a finite union of sets of the form $t_i\times (-\infty,0]$ at times $t_i$ and either $\emptyset$ or $\mathbb{R}$ for each $t$ in between $t_{i-1}$ and $t_i$. 

Now, if in such a union the set $(t_{i-1},t_i)\times \mathbb{R}$ is present at least once, then the reachability probability in the limit equals exactly one. This is because any trajectory $\omega$ will enter the set $R$ in that subregion. In this case, convergence is easily shown. In fact, being the Skorokhod space a Polish space, 
any converging sequence  $G^N\Rightarrow G$ of random variables in that space is uniformly tight, meaning that for each $\epsilon$ there is a compact space $K_\epsilon$ such that, outside it, all random variables and the limit have probability less than $\epsilon$. In particular, a compact set of trajectories is bounded in $[0,T]$ with respect to the sup norm  \cite{adler2010geometry}, meaning that for each $\epsilon$ there is a $k_\epsilon>0$ such that the probability that a trajectory $\omega$ has modulus $|\omega(t)|\leq k$ uniformly in $[0,T]$ is more than $1-\epsilon$ for all $N$. Now, consider the time interval $(t_{i-1},t_i)$ where the reachable set converges to  $(t_{i-1},t_i)\times \mathbb{R}$ in the limit. As the threshold $b^N(t)$ is an analytic function of $t$, removing a region of length $\Delta $ near $t_{i-1}$ and $t_i$ (i.e. restricting to $[t_{i-1}+\Delta ,t_i-\Delta ]$), we can find an $N_0$ such that, for $N>N_0$, $b^N(t)$ is greater than $k_\epsilon$ uniformly in $[t_{i-1}+\Delta ,t_i-\Delta ]$. Then the limit of the reachability for $G^N$ is greater than $1-\epsilon$ for any epsilon, that is, it equals one. The case in which the limit region $R$ is the empty set for every $t$ is easily proved along the same lines.

The interesting case is the one in which there are some $t_i$'s where $b^N(t_i) = 0$ for all $N$, but it is always negative outside them, implying the reachable region $R$ converges to the empty set everywhere but in the $t_i$'s, where it equals $(-\infty,0]$. This corresponds to the scenario in which the fluid limit $\Phi(t)$ is tangent to the reachable set, but never enters it, a scenario known to cause trouble in the use of mean field to estimate hitting times \cite{bortolussi2016hybrid}.

To deal with this last case, let us denote with $P^N$ the probability in the trajectory space for $B G^N$, and with $P$ the probability for $B G$. 

As before, denote with $R_N$ the reachability set for $G^N$ and with $R$ the limit set, taking the threshold $b^N$ to infinity. We now introduce a set which over-approximates $R_N$ for $N$ large. This set is defined as follows:  invoking uniform tightness, we fix a large value $k_\epsilon$ as before, so that trajectories of $B G^N$ and of $B G$ are contained in $[-k_\epsilon,k_\epsilon]$ with probability $1-\epsilon$, uniformly for $t \in [0, T]$. Furthermore, we consider points $t_i$ where $b^N(t_i)$ is zero, and take a small neighborhood $B^\Delta _i$ of width $\Delta $ around them. Define the set $R_\epsilon$ in the trajectory space as:

\begin{align*}
    R_{\epsilon} = \{\omega(t)|\omega(t) \leq 0,  \text{for }t \in B^\Delta _i, \omega(t) \leq - k_\epsilon\ \text{elsewhere in }[0, T]\}.
\end{align*}

By relying on the continuity of the set $R$ for $G$, we can choose $\Delta $ small enough so as to enforce that $|P (R_\epsilon) - P (R)| \leq \epsilon$. The continuity of $R$ for $G$ follows from the fact that $\omega\in R$ if and only if $\omega(t_i)\leq 0$ for $i=1,\ldots,n$, i.e. $R$ is a finite dimensional projection on $t_i$'s. Therefore, its boundary is a set of topological dimension less than $n$ in $\mathbb{R}^n$, which has probability zero under the finite dimensional projection of $G$ on $t_i$'s (which is Gaussian).  Now, using triangular inequality, we get:

\begin{eqnarray*}
|P^N(R_N) - P(R) | &\leq & |P^N(R_N) - P^N(R)| + |P^N(R) - P(R)| \\ 
& \leq & |P^N(R_\epsilon) - P^N(R)| + |P^N(R) - P(R)| \\
& \leq & |P^N(R _\epsilon) - P(R _\epsilon)| \\
& + & |P(R) - P(R _\epsilon)| + 2 |P^N(R) - P(R)|.
\end{eqnarray*}
The second inequality above follows from the monotonic behaviour of probability distributions, as for each $\Delta $ and $k_\epsilon$ there is an $N_0$ such that, for all $N\geq N_0$, $R\subset R_N\subset R_\epsilon$, hence $|P^N(R^N) - P^N(R)| \leq   |P^N(R_{\epsilon}) - P^N(R)|$. 

Furthermore, $|P^N(R) - P(R)| \rightarrow 0$, by the continuity of the set $Y$. In $R$, by virtue of Lemmas 1 and 2 below, if also follows that $|P^N(R _\epsilon) - P(R _\epsilon)| \rightarrow 0$, and hence:

\[\limsup_{N\rightarrow \infty} |P^N(R_N) - P(R) | \leq \epsilon, \]
which holds for any $\epsilon>0$,  allowing us to conclude that:
\[\lim_{N\rightarrow \infty} |P^N(R_N) - P(R) | =0, \]
as desired.

\medskip

\noindent \textbf{Lemma 1.} Let $b\in \mathbb{R}$, and consider the reachable set $R = \{\omega|\exists t:\omega(t)\leq b\}$. Then $P^N(R)\rightarrow P(R)$, with $P^N$, $P$ as above. 

\noindent \textbf{Proof.}  The boundary of the reachable set $R$ is the set of trajectories $\omega$ such that $\inf_{t\in[0,T]}\omega(t) = b$. In order to conclude, we need to show that this set has measure $0$. As $G$ is a Gaussian process, assuming the covariance function is non-zero, we have that the distribution of the infimum (or equivalently the supremum) is absolutely continuous \cite{lifshits1984absolute}, which implies that the set of trajectories for which $ inf_{t\in [0,T]} \omega(t)=b$ has measure $0$. Hence, $R$ is a continuity set for $G$, which prove the thesis due to the Portmanteau theorem.

\medskip

\noindent \textbf{Lemma 2.} Consider a reachable set $R$ defined by a piecewise constant threshold. Hence, fix $0=t_1,\ldots t_{n+1}=T \in [0,T]$, and $b_i \in \mathbb{R}$, for $i=1,\ldots,n$, and let $R = \{\omega|\exists i \in\{1,\ldots,n\}, \exists t\in [t_i,t_{i+1}]:\omega(t)\leq b_i\}$. Then $P^N(R)\rightarrow P(R)$, with $P^N$, $P$ as above. 

\noindent \textbf{Proof.} We proceed by induction on $j$, showing that $R$ is a continuity set for $G$. The case for $j=1$ follows from Lemma 1 above. Suppose we proved it up to $j-1$. Then, conditioned on an initial trajectory $\omega$ from time zero to $t_j$, with $\omega(t_j) = y$, $G$ restricted in $[t_j,t_{j+1}]$ is a Gaussian process, and we can apply Lemma 1 to show that the probability of $\partial R$, restricted in this time span, is zero. Now, the probability of $\partial R$ restricted to $[0,t_{j+1}]$ can be bounded by the sum of two terms. The first is the probability of $\partial R$ in $[0,t_{j}]$ , which is zero, the second is  probability of $\partial R \cup R^c$ up to time $t_j$ times the probability of $\partial R$ in $[t_j,t_{j+1}]$, conditional on being in $\partial R \cup R^c$ up to time $t_j$. Also this second term is zero, as the conditional probability is zero for any initial trajectory $\omega$. The bound on the probability of $\partial R$ follows because any trajectory in $\partial R$ up to time $t_{j+1}$ is either touching $b_j$ between $[t_j,t_{j+1}]$ (second term), or before $t_j$ (first term). The second case overlaps with the first for all trajectories  that touch the threshold both before $t_j$  and between $[t_j,t_{j+1}]$.

\bigskip

\noindent\textbf{Proposition \ref{PropositionBoundedRewa}.}
For $T\in \mathbb{R}_{\geq 0}$ and $B\in\mathbb{R}^{|\Lambda|\times k}$ let $A$ be the set defined as $A=\{x\in \mathbb{R}^{|\Lambda|} \,s.t.\, \forall i \in \{1,..,m \}, (Bx)_i\leq b_i \}$. Then, for $\hat x_0 \in \mathbb{R}^{|\Lambda|}$ and $z_{d,0}$, the state in the state space of $\hat Z^{\Delta z,h,N}$ corresponding to the region containing $\hat x_0,$ it holds that:
$$ \lim_{N\to \infty}\lim_{h\to 0}\lim_{\Delta z \to 0}|\rho_{reach}(\hat x_0,\hat X^N,T,A)- \rho_{reach}(z_{d,0}, \hat Z^{\Delta z,h,N},\lfloor \frac{T}{h}  \rfloor,A)|=0.$$

\noindent\textbf{Proof.} 
In order to prove the convergence, we start by introducing some notation. First of all,  $B \hat{X}^N$ and $B \hat{Y}^N$ are the CTMC and its CLA projected on the inequalities defining the region $A$.
%In the following, we are gin probabvavvavvvaavvoing to work always with these projected versions of $\hat{X}^N$ and $\hat{Y}^N$. As such, we will omit the term $B$ to avoid overhead of notation. Hence, whenever we write $\hat{X}^N$ we actually intend $B \hat{X}^N$, and so on.  
Additionally we denote by  $\hat {Z}^{h,N}$  the DTMP obtained by time discretization of $B \hat{Y}^N$, and $\hat Z^{\Delta z,h,N}$ is the space discretization of $\hat {Z}^{h,N}$.

We now introduce the following stopping times, which are random variables on $\mathbb{R}_{\geq 0}$ denoting the random time in which a certain event happens. In particular, we are interested in the stopping times corresponding to the event of entering into the region $A$, usually known as hitting times, for the different processes we consider:
\begin{itemize}
    \item $\T^N$ is the hitting time for $\hat X^N$;
    \item $\bar\T^N$ is the hitting time for $\hat Y^N$;
    \item $\T^{N,h}$ is the hitting time for $\hat {Z}^{h,N}$;
    \item $\T^{N,h,\Delta z}$ is the hitting time for $\hat Z^{\Delta z,h,N}$.
\end{itemize}
Hitting times are strictly related to the reachability probability. For instance, $Prob\{\exists t\leq T: \hat X^N(t) \in A \} = Prob\{\T^N \leq T\}$. 
Furthermore, we introduce also the stopping time $\T^G$, which is the 
hitting time for the Gaussian process $G(t)$ to enter the rescaled  region $A^\infty$, which is the limiting region, similarly to what we do in the proof of the Theorem \ref{COnvergenceReach}. 
We have the following weak convergence relationships for such hitting times:
\begin{itemize}
    \item $\T^{N,h,\Delta z}\Rightarrow\T^{N,h}$ as $\Delta z\rightarrow 0$;
    \item $\T^{N,h}\Rightarrow\bar\T^{N}$ as $h\rightarrow 0$;
    \item $\bar\T^{N}\Rightarrow\T^G$ as $N\rightarrow \infty$;
    \item $\T^{N}\Rightarrow\T^G$ as $N\rightarrow \infty$;
\end{itemize}
To show these relationships, one just has to use the correspondence of hitting times with the reachability probability, and the convergence of the latter by virtue of the proof of Theorem 4. For instance 
$Prob\{\T^N \leq T\} = Prob\{\exists t\leq T: \hat X^N(t) \in A \} \rightarrow_{N\rightarrow\infty} Prob\{\exists t\leq T: G(t) \in A^\infty \} = Prob\{\T^G \leq T\}$. The pointwise convergence of the cumulative distributions function of $\T^N$ to that  of $\T$ implies weak convergence by the Portmanteau theorem \cite{billingsley2013convergence}.

In order to prove the convergence of rewards, given a reward structure $\rho$ on $\mathbb{R}^m$ and a path  $\omega:\mathbb{R}_{\geq 0}\to \mathbb{R}^{m},m>0,$ we define the functional $\mathcal{R}(\omega,T)=\int_{0}^{T}\rho(\omega(s)))ds$. In order to evaluate the desired reward, we need to stop the integration as soon as the process enters the target region $A$, hence 
$\rho_{reach}(\hat X^N,T,A) = \E[\mathcal{R}(\hat X^N,\T^N)]$, where the expectation is taken with respect to both $X^N$ and $\T^N$. 
% Call $B \frac{\hat{X}^N}{N}$ the modified process that is equal to $B \frac{X^N}{N}$ except for states in $A$ that are made absorbing.
Then, by triangular inequality, we have:
\begin{align*}
    |\mathbb{E}[\mathcal{R}(B\hat{X}^N,\T^N)]-& \mathbb{E}[\mathcal{R}(\hat Z^{\Delta z,h,N},\T^{N,h,\Delta z})]|\leq\\ 
    & |\mathbb{E}[\mathcal{R}(B \hat{X}^N,\T^{N})]- \mathbb{E}[\mathcal{R}(B \hat{Y}^N,\bar\T^{N})]|+\\
   &|\mathbb{E}[\mathcal{R}(B \hat{Y}^N,\bar\T^{N})]- \mathbb{E}[\mathcal{R}(\hat {Z}^{h,N},\T^{N,h})]|+\\&
   |\mathbb{E}[\mathcal{R}(\hat {Z}^{h,N},\T^{N,h})]- \mathbb{E}[\mathcal{R}(\hat Z^{\Delta z,h,N},\T^{N,h,\Delta z})]|.
\end{align*}
% where $\hat {Z}^{h,N}$ is the DTMP obtained by time discretization of $B\hat{Y}^N$, and $\hat Z^{\Delta z,h,N}$ is the space discretization of $\hat {Z}^{h,N}$.

We will prove the proposition by showing that all three terms on the right hand side of the above inequality converge to zero. In particular, the third term can be sent to zero for only $\Delta z\rightarrow 0$, and the second term by sending only $h\rightarrow 0$, as both are related to the discretization of  $B\hat Y^N$. Instead, the first term depends only on $N$. 

We will start with the second term. First, results in \cite{laurenti2017reachability} imply that $\hat Z^{h,N} \rightarrow B \hat Y^N$ in probability as $h\rightarrow 0$. Furthermore,  
Theorem \ref{Time Discretization} gives us weak convergence of the hitting times: $\T^{N,h}\Rightarrow \bar\T^N$. 
The challenge in the second term lies in the fact that it depends on two random variables, so we need to rely again on triangular inequality to separate them:  
\begin{align*}
|\mathbb{E}[\mathcal{R}(B \hat{Y}^N,\bar\T^{N})] &- \mathbb{E}[\mathcal{R}(\hat {Z}^{h,N},\T^{N,h})]| \leq\\ 
& |\mathbb{E}[\mathcal{R}(B \hat{Y}^N,\bar\T^{N})]- \mathbb{E}[\mathcal{R}(\hat {Z}^{h,N},\bar\T^{N})]| +\\
& |\mathbb{E}[\mathcal{R}(\hat {Z}^{h,N},\bar\T^{N})]- \mathbb{E}[\mathcal{R}(\hat {Z}^{h,N},\T^{N,h})]|
\end{align*}

Consider now a term appearing in the right hand side, e.g. $\mathbb{E}[\mathcal{R}(B \hat{Y}^N,\bar\T^{N})]$. As the expectation is taken with respect to both $B \hat{Y}^N$ and $\bar\T^N$, we can rely on the following conditional expectation decomposition:
\[\mathbb{E}_{B \hat{Y}^N,\bar\T^N}[\mathcal{R}(B \hat{Y}^N,\bar\T^{N})] = \mathbb{E}_{\bar\T^N}[ \mathbb{E}_{B \hat{Y}^N}[\mathcal{R}(B \hat{Y}^N,t)~|~\bar\T^{N}=t]]. \]
Furthermore, recall that:
\[ \mathbb{E}_{B \hat{Y}^N}[\mathcal{R}(B \hat{Y}^N,t)] = \int_0^t \mathbb{E}[\rho(B \hat{Y}^N(s)ds]. \]
Now, consider the term $|\mathbb{E}[\mathcal{R}(B \hat{Y}^N,\bar\T^{N})]- \mathbb{E}[\mathcal{R}(\hat {Z}^{h,N},\bar\T^{N})]|$. Applying the previous decomposition, we can upper bound it by:
\[\mathbb{E}_{\bar\T^N}\left[ \int_0^t |\mathbb{E}[\rho(B \hat{Y}^N(s)~|~\bar\T^{N}=t] - \mathbb{E}[\rho(\hat {Z}^{h,N}(s)~|~\bar\T^{N}=t]|ds  \right], \]
where we assume that $\hat {Z}^{h,N}(s)$ is a piecewise constant function in between each step at distance $h$, to write its cumulative reward as an integral. 

We have that $\sup_{s\leq t}|B \hat{Y}^N(s) -  \hat {Z}^{h,N}(s)| $ converges to zero in probability as $h\rightarrow 0$ \cite{laurenti2017reachability}.  From this, we can deduce that  
$\E[\sup_{s\leq t}|B \hat{Y}^N(s) -  \hat {Z}^{h,N}(s)|] $ converges to zero. This proof presented in \cite{laurenti2017reachability} is consequence of the \emph{Borell-TIS inequality} \cite{adler2010geometry}, which guarantees that the supremum of a Gaussian process is still normally distributed.  %\todo{Sounds reasonable, but I am not sure. Luca L: I leave this to you, as this is related to the time discretization you  master! we can also show that $\sup_{s\leq t} \E[|B \hat{Y}^N(s) -  \hat {Z}^{h,N}(s)|]$ goes to zero if easier.}
Hence:
$|\mathbb{E}[\rho(B \hat{Y}^N(s)~|~\bar\T^{N}=t] - \mathbb{E}[\rho(\hat {Z}^{h,N}(s)~|~\bar\T^{N}=t]| \leq \mathbb{E}[|\rho(B \hat{Y}^N(s)) - \rho(\hat {Z}^{h,N}(s))|~\bar\T^{N}=t] \leq L_\rho \mathbb{E}[|B \hat{Y}^N(s) - \hat {Z}^{h,N}(s)|~\bar\T^{N}=t] \leq L_\rho  \mathbb{E}[\sup_{s\leq T}|B \hat{Y}^N(s) - \hat {Z}^{h,N}(s)|~\bar\T^{N}=t]= L_\rho \Delta _h$, which converges to zero by the discussion above. Recall that in the above $L_\rho$ is the Lipschitz constant of reward $\rho$ . 
Hence we can bound the first term by $\E[\int_0^t\Delta _hds] \leq \Delta _h T$, which goes to zero as $h\rightarrow 0$. 

Consider now the term $|\mathbb{E}[\mathcal{R}(\hat {Z}^{h,N},\bar\T^{N})]- \mathbb{E}[\mathcal{R}(\hat {Z}^{h,N},\T^{N,h})]|$: it tends to zero by application of the Portmanteau theorem, owing to the weak convergence of $\T^{N,h}$ to $\bar\T^{N}$, and the fact that $\mathcal{R}(\hat {Z}^{h,N},t)$ is a bounded and continuous function of $t$ (being the cumulative reward up to time $t$ of a bounded function $\rho$). 

The third term in the main inequality, $|\mathbb{E}[\mathcal{R}(\hat {Z}^{h,N},\T^{N,h})]- \mathbb{E}[\mathcal{R}(\hat Z^{\Delta z,h,N},\T^{N,h,\Delta z})]|$, can be shown to converge to zero using a similar approach, owing to the convergence of the space discretization to the DTMP $\hat {Z}^{h,N}$, and the convergence of the hitting times. 

What is left is the first term of the main inequality of the theorem, namely $|\mathbb{E}[\mathcal{R}(B \hat{X}^N,\T^{N})]- \mathbb{E}[\mathcal{R}(B \hat{Y}^N,\bar\T^{N})]|$, which has to converge to zero as $N$ diverges. 

To simplify the notation below, let us define: 
\begin{itemize}
    \item $g^N(t) = \mathbb{E}[\mathcal{R}(B \hat{X}^N,t)]$ is the cumulative reward for $B \hat{X}^N$ up to time $t$
    \item $\gamma^N(t) = \mathbb{E}[\mathcal{R}(B \hat{Y}^N,t)]$ is the cumulative reward for $B \hat{Y}^N$ up to time $t$.
\end{itemize}
Then the first term can be bounded by:
\begin{eqnarray}
\|\E[g^N(\T^N)] - \E[\gamma^N(\bar \T^N)]\| 
% & \leq & \|E[g^N(T^N)] - E[g^\infty(T)]\|\nonumber\\
% & + & \|E[\gamma^\infty( T)] - E[\gamma^N(\bar T^N)]\|\nonumber\\
& \leq & \|\E[g^N(\T^N)] - \E[g^\infty(\T^N)]\|\nonumber\\
& + & \|\E[\gamma^\infty(\T^N)] - \E[\gamma^\infty(\T)]\|\nonumber\\
& + & \|\E[\gamma^\infty(\T)] - \E[\gamma^\infty(\bar \T^N)]\|\nonumber\\
& + & \|\E[\gamma^\infty(\bar \T^N)] - \E[\gamma^N(\bar \T^N)]\|\nonumber
\end{eqnarray}
where $g^\infty = \gamma^\infty$ is the cumulative reward for the fluid limit $B\hat X^\infty = B\Phi$.

Consider the first term in the above inequality: 
\begin{eqnarray*}
\|\E[g^N(\T^N)] - \E[g^\infty(\T^N)]\| & \leq & \E_{t\sim \T^N}\left[ \E\left[\int_0^{t} \| \rho(\hat X^N(s)) - \rho(\Phi(s))ds\|\right]\right]\\
& \leq & \E_{t\sim \T^N}\left[ \int_0^{t}  L_\rho \E[\|X^N(s) - \Phi(s)\|] \right]\\ 
& \leq & \E_{t\sim \T^N}\left[ \int_0^t L_\rho \sup_{s\leq T} \E[\|  X^N(s) - \Phi(s)\|] \right].
\end{eqnarray*}
Now 
$ \sup_{s\leq T} E[\|  X^N(s) - \Phi(s)\|]$ converges to zero by virtue of a corollary of the fluid approximation theorem on the rate of convergence of expectations \cite{Gast2017}, meaning that there is $N_1$ such that, for $N\geq N_1$, it is less than $\epsilon/(4 T)$. For all such $N$, it follows that $\|E[g^N(\T^N)] - E[g^\infty(\T^N)]\| \leq \epsilon/4$.

Le us deal with the fourth term: \[\|\E[\gamma^\infty(\bar \T^N)] - \E[\gamma^N(\bar \T^N)]\|\leq \E_{t\sim \bar \T^N}[\|\gamma^\infty(t)- \gamma^N(t)\|].\]
For a fixed $t$, we have that $\|\gamma^\infty(t)- \gamma^N(t)\| \leq \int_0^t \E[\|\rho(\Phi(s) + G(s)/\sqrt{N}) -\rho(\Phi(s))\| ] \leq \int_0^t \E[L_\rho\|G(s)/\sqrt{N}\| ] = L_\rho \int_0^t \E[\sup_{s\leq T}|G(s)|]/\sqrt{N}$. Now, as $G(t)$ has bounded convariance matrix in $[0,T]$, $\E[\sup_{s\leq T}|G(s)|]$ is finite, say equal to $M_G$, hence  $\|\gamma^\infty(t)- \gamma^N(t)\| \leq L_\rho M_G t / \sqrt{N}$, and so $\|\E[\gamma^\infty(\bar \T^N)] - E[\gamma^N(\bar \T^N)]\|\leq L_\rho M_G T / \sqrt{N}$ which is less than $\epsilon/4$ for $N\geq N_4$, for some $N_4>0$. 

Terms two and three in the inequality above, instead, converge by virtue of the Portmanteau theorem and of the weak convergence of $\T^N$ or $\bar \T^N$ to $\T^G$, hence there is $N_2$ such that they are less that $\epsilon/4$ for $N\geq N_2$. It then follows that:
\[\limsup_{N\rightarrow\infty} \|\E[g^N(\T^N)] - \E[\gamma^N(\bar \T^N)]\| < \epsilon \]
for an arbitrary $\epsilon$, implying:
\[\lim_{N\rightarrow\infty} \|E[g^N(\T^N)] - E[\gamma^N(\bar \T^N)]\| = 0. \]

Thus, we showed that $|\mathbb{E}[\mathcal{R}(B\hat{X}^N,\T^N)]- \mathbb{E}[\mathcal{R}(\hat Z^{\Delta z,h,N},\T^{N,h,\Delta z})]|$ converges to zero for $\Delta z,h$ tending to zero and $N$ diverging, as so do all the three terms bounding it.

\vspace{-1em}
\bibliographystyle{abbrv}
{\scriptsize
\bibliography{Biblio}

\begin{thebibliography}{10}

\bibitem{abate2010approximate}
A.~Abate, J.-P. Katoen, J.~Lygeros, and M.~Prandini.
\newblock Approximate model checking of stochastic hybrid systems.
\newblock {\em European Journal of Control}, 16(6):624--641, 2010.

\bibitem{adler2010geometry}
R.~J. Adler.
\newblock {\em The geometry of random fields}.
\newblock SIAM, 2010.

\bibitem{anderson2015models}
D.~F. Anderson and T.~G. Kurtz.
\newblock {\em Stochastic analysis of biochemical systems}.
\newblock Springer.

\bibitem{Anderson2011}
D.~F. Anderson and T.~G. Kurtz.
\newblock Continuous time {M}arkov chain models for chemical reaction networks.
\newblock In {\em Design and analysis of biomolecular circuits}, pages 3--42.
  Springer, 2011.

\bibitem{Angluin2008}
D.~Angluin, J.~Aspnes, and D.~Eisenstat.
\newblock A simple population protocol for fast robust approximate majority.
\newblock {\em Distributed Computing}, 21(2):87--102, 2008.

\bibitem{aziz1996verifying}
A.~Aziz, K.~Sanwal, V.~Singhal, and R.~Brayton.
\newblock Verifying continuous time {Markov} chains.
\newblock In {\em Computer Aided Verification}, pages 269--276. Springer, 1996.

\bibitem{aziz2000model}
A.~Aziz, K.~Sanwal, V.~Singhal, and R.~Brayton.
\newblock Model-checking continuous-time markov chains.
\newblock {\em ACM Transactions on Computational Logic (TOCL)}, 1(1):162--170,
  2000.

\bibitem{baier2000model}
C.~Baier, B.~Haverkort, H.~Hermanns, and J.-P. Katoen.
\newblock Model checking continuous-time markov chains by transient analysis.
\newblock In {\em International Conference on Computer Aided Verification},
  pages 358--372. Springer, 2000.

\bibitem{baier2003model}
C.~Baier, B.~Haverkort, H.~Hermanns, and J.-P. Katoen.
\newblock Model-checking algorithms for continuous-time markov chains.
\newblock {\em Software Engineering, IEEE Transactions on}, 29(6):524--541,
  2003.

\bibitem{baier2008principles}
C.~Baier, J.-P. Katoen, et~al.
\newblock {\em Principles of model checking}, volume 26202649.
\newblock MIT press Cambridge, 2008.

\bibitem{bertsekas2004stochastic}
D.~P. Bertsekas and S.~Shreve.
\newblock {\em Stochastic optimal control: the discrete-time case}.
\newblock 2004.

\bibitem{billingsley2013convergence}
P.~Billingsley.
\newblock {\em Convergence of probability measures}.
\newblock John Wiley \& Sons, 2013.

\bibitem{bortolussi2016hybrid}
L.~Bortolussi.
\newblock Hybrid behaviour of markov population models.
\newblock {\em Information and Computation}, 247:37--86, 2016.

\bibitem{bortolussi2016approximation}
L.~Bortolussi, L.~Cardelli, M.~Kwiatkowska, and L.~Laurenti.
\newblock Approximation of probabilistic reachability for chemical reaction
  networks using the linear noise approximation.
\newblock In {\em International Conference on Quantitative Evaluation of
  Systems}, pages 72--88. Springer, 2016.

\bibitem{bortolussi2012fluid}
L.~Bortolussi and J.~Hillston.
\newblock Fluid model checking.
\newblock In {\em CONCUR 2012--Concurrency Theory}, pages 333--347. Springer,
  2012.

\bibitem{bortolussi2015efficient}
L.~Bortolussi and J.~Hillston.
\newblock Efficient checking of individual rewards properties in markov
  population models.
\newblock {\em arXiv preprint arXiv:1509.08561}, 2015.

\bibitem{bortolussi2013continuous}
L.~Bortolussi, J.~Hillston, D.~Latella, and M.~Massink.
\newblock Continuous approximation of collective system behaviour: {A}
  tutorial.
\newblock {\em Performance Evaluation}, 70(5):317--349, 2013.

\bibitem{bortolussi2013model}
L.~Bortolussi and R.~Lanciani.
\newblock Model checking {M}arkov population models by central limit
  approximation.
\newblock In {\em Quantitative Evaluation of Systems}, pages 123--138.
  Springer, 2013.

\bibitem{bortolussi2014stochastic}
L.~Bortolussi and R.~Lanciani.
\newblock Stochastic approximation of global reachability probabilities of
  {M}arkov population models.
\newblock In {\em Computer Performance Engineering}, pages 224--239. Springer,
  2014.

\bibitem{bortolussi2016smoothed}
L.~Bortolussi, D.~Milios, and G.~Sanguinetti.
\newblock Smoothed model checking for uncertain continuous-time {Markov}
  chains.
\newblock {\em Information and Computation}, 2016.

\bibitem{cardelli2017syntax}
L.~Cardelli, M.~{\v{C}}e{\v{s}}ka, M.~Fr{\"a}nzle, M.~Kwiatkowska, L.~Laurenti,
  N.~Paoletti, and M.~Whitby.
\newblock Syntax-guided optimal synthesis for chemical reaction networks.
\newblock In {\em International Conference on Computer Aided Verification},
  pages 375--395. Springer, 2017.

\bibitem{cardelli2016stochastic}
L.~Cardelli, M.~Kwiatkowska, and L.~Laurenti.
\newblock Stochastic analysis of chemical reaction networks using linear noise
  approximation.
\newblock {\em Biosystems}, 149:26--33, 2016.

\bibitem{cardelli2016stochasticHybrid}
L.~Cardelli, M.~Kwiatkowska, and L.~Laurenti.
\newblock A stochastic hybrid approximation for chemical kinetics based on the
  linear noise approximation.
\newblock In {\em International Conference on Computational Methods in Systems
  Biology}, pages 147--167. Springer, 2016.

\bibitem{chen2009ltl}
T.~Chen, T.~Han, J.-P. Katoen, and A.~Mereacre.
\newblock Ltl model checking of time-inhomogeneous markov chains.
\newblock {\em Automated Technology for Verification and Analysis}, pages
  104--119, 2009.

\bibitem{ciocchetta2009bio}
F.~Ciocchetta and J.~Hillston.
\newblock Bio-pepa: A framework for the modelling and analysis of biological
  systems.
\newblock {\em Theoretical Computer Science}, 410(33-34):3065--3084, 2009.

\bibitem{Dannenberg:2015:CCR:2737798.2688907}
F.~Dannenberg, E.~M. Hahn, and M.~Kwiatkowska.
\newblock Computing cumulative rewards using fast adaptive uniformization.
\newblock {\em ACM Trans. Model. Comput. Simul.}, 25(2):9:1--9:23, Feb. 2015.

\bibitem{DHK14}
F.~Dannenberg, E.~M. Hahn, and M.~Kwiatkowska.
\newblock Computing cumulative rewards using fast adaptive uniformization.
\newblock volume~25, page~9. ACM, 2015.

\bibitem{ethier2009markov}
S.~N. Ethier and T.~G. Kurtz.
\newblock {\em Markov processes: characterization and convergence}, volume 282.
\newblock John Wiley \& Sons, 2009.

\bibitem{Gast2017}
N.~Gast.
\newblock Expected values estimated via mean-field approximation are
  1/n-accurate.
\newblock {\em Proc. ACM Meas. Anal. Comput. Syst.}, 1(1):17:1--17:26, 2017.

\bibitem{gillespie1977exact}
D.~T. Gillespie.
\newblock Exact stochastic simulation of coupled chemical reactions.
\newblock {\em The journal of physical chemistry}, 81(25):2340--2361, 1977.

\bibitem{gillespie1992rigorous}
D.~T. Gillespie.
\newblock A rigorous derivation of the chemical master equation.
\newblock {\em Physica A: Statistical Mechanics and its Applications},
  188(1):404--425, 1992.

\bibitem{grima2015linear}
R.~Grima.
\newblock Linear-noise approximation and the chemical master equation agree up
  to second-order moments for a class of chemical systems.
\newblock {\em Physical Review E}, 92(4):042124, 2015.

\bibitem{heath2008probabilistic}
J.~Heath, M.~Kwiatkowska, G.~Norman, D.~Parker, and O.~Tymchyshyn.
\newblock Probabilistic model checking of complex biological pathways.
\newblock {\em Theoretical Computer Science}, 391(3):239--257, 2008.

\bibitem{hillston2005compositional}
J.~Hillston.
\newblock {\em A compositional approach to performance modelling}, volume~12.
\newblock Cambridge University Press, 2005.

\bibitem{karp1969parallel}
R.~M. Karp and R.~E. Miller.
\newblock Parallel program schemata.
\newblock {\em Journal of Computer and system Sciences}, 3(2):147--195, 1969.

\bibitem{krantz_realanalytic_2002}
S.~Krantz and P.~Harold.
\newblock {\em A {Primer} of {Real} {Analytic} {Functions} ({Second} ed.)}.
\newblock Birkḧauser, 2002.

\bibitem{Kwiatkowska2007}
M.~Kwiatkowska, G.~Norman, and D.~Parker.
\newblock Stochastic model checking.
\newblock In {\em Formal methods for performance evaluation}, pages 220--270.
  Springer, 2007.

\bibitem{kwiatkowska2011prism}
M.~Kwiatkowska, G.~Norman, and D.~Parker.
\newblock Prism 4.0: Verification of probabilistic real-time systems.
\newblock In {\em Computer aided verification}, pages 585--591. Springer, 2011.

\bibitem{laurenti2017reachability}
L.~Laurenti, A.~Abate, L.~Bortolussi, L.~Cardelli, M.~Ceska, and
  M.~Kwiatkowska.
\newblock Reachability computation for switching diffusions: Finite
  abstractions with certifiable and tuneable precision.
\newblock In {\em Proceedings of the 20th International Conference on Hybrid
  Systems: Computation and Control}, pages 55--64. ACM, 2017.

\bibitem{lifshits1984absolute}
M.~Lifshits.
\newblock Absolute continuity of functionals of “supremum” type for
  gaussian processes.
\newblock {\em Journal of Mathematical Sciences}, 27(5):3103--3112, 1984.

\bibitem{maler2004monitoring}
O.~Maler and D.~Nickovic.
\newblock Monitoring temporal properties of continuous signals.
\newblock In {\em Formal Techniques, Modelling and Analysis of Timed and
  Fault-Tolerant Systems}, pages 152--166. Springer, 2004.

\bibitem{milios2017probabilistic}
D.~Milios, G.~Sanguinetti, and D.~Schnoerr.
\newblock Probabilistic model checking for continuous time markov chains via
  sequential bayesian inference.
\newblock {\em arXiv preprint arXiv:1711.01863}, 2017.

\bibitem{munsky2006finite}
B.~Munsky and M.~Khammash.
\newblock The finite state projection algorithm for the solution of the
  chemical master equation.
\newblock {\em The Journal of chemical physics}, 124(4):044104, 2006.

\bibitem{murata1989petri}
T.~Murata.
\newblock Petri nets: Properties, analysis and applications.
\newblock {\em Proceedings of the IEEE}, 77(4):541--580, 1989.

\bibitem{pnueli1977temporal}
A.~Pnueli.
\newblock The temporal logic of programs.
\newblock In {\em Foundations of Computer Science, 1977., 18th Annual Symposium
  on}, pages 46--57. IEEE, 1977.

\bibitem{schilling2017measures}
R.~L. Schilling.
\newblock {\em Measures, integrals and martingales}.
\newblock Cambridge University Press, 2017.

\bibitem{schnoerr2017efficient}
D.~Schnoerr, B.~Cseke, R.~Grima, and G.~Sanguinetti.
\newblock Efficient low-order approximation of first-passage time
  distributions.
\newblock {\em Physical review letters}, 119(21):210601, 2017.

\bibitem{shahrezaei2008analytical}
V.~Shahrezaei and P.~S. Swain.
\newblock Analytical distributions for stochastic gene expression.
\newblock {\em Proceedings of the National Academy of Sciences},
  105(45):17256--17261, 2008.

\bibitem{thattai2001}
M.~Thattai and A.~Van~Oudenaarden.
\newblock Intrinsic noise in gene regulatory networks.
\newblock {\em Proceedings of the National Academy of Sciences},
  98(15):8614--8619, 2001.

\bibitem{Kampen1992b}
N.~G. Van~Kampen.
\newblock {\em Stochastic processes in physics and chemistry}, volume~1.
\newblock Elsevier, 1992.

\bibitem{Wallace2012}
E.~Wallace, D.~Gillespie, K.~Sanft, and L.~Petzold.
\newblock Linear noise approximation is valid over limited times for any
  chemical system that is sufficiently large.
\newblock {\em IET systems biology}, 6(4):102--115, 2012.

\bibitem{wolf2010solving}
V.~Wolf, R.~Goel, M.~Mateescu, and T.~A. Henzinger.
\newblock Solving the chemical master equation using sliding windows.
\newblock {\em BMC systems biology}, 4(1):1, 2010.

\end{thebibliography}
}

\end{document}